\documentclass[11pt]{article}
\usepackage[letterpaper,margin=1in]{geometry}
\usepackage{amsmath,amsthm,amssymb}
\usepackage{graphicx}
\usepackage{algorithmic}
\usepackage{enumerate}
\usepackage{latexsym}
\usepackage{caption}
\usepackage{multirow}
\usepackage{paralist}
\usepackage[normalem]{ulem}
\usepackage{authblk}
\usepackage[dvipsnames]{xcolor}

\usepackage{booktabs} % For formal tables
\usepackage[noend,ruled,noline]{algorithm2e} % For algorithms
\DontPrintSemicolon
\usepackage{enumitem}
\usepackage{mathtools}
\usepackage{amsmath,amsthm}
\SetAlFnt{\small}
\SetAlCapFnt{\small}
\SetAlCapNameFnt{\small}
\SetAlCapHSkip{0pt}
\IncMargin{-\parindent}

\usepackage{natbib}

\usepackage{sectsty}
\sectionfont{\fontsize{13}{15}\selectfont}

 \newtheorem{theorem}{Theorem}
 
  \newtheorem{lemma}[theorem]{Lemma}

  \newtheorem*{definition*}{Definition}
  
  \newtheorem{corollary}[theorem]{Corollary}
  
  \theoremstyle{remark}
 
   \usepackage{thm-restate}

\usepackage{tikz}
\usepackage{verbatim}
\usetikzlibrary{shapes,arrows,fit,calc,positioning}
\usetikzlibrary{arrows}
\usetikzlibrary{decorations.markings}
\usetikzlibrary{decorations.text}

    \def \N {\mathbb{N}}
    \def \R {\mathbb{R}}

\usepackage{tikz}
\usetikzlibrary{decorations.markings, arrows.meta, positioning}
\tikzset{>={Latex[width=2mm,length=2mm]}}
\definecolor{cobalt}{rgb}{0.0, 0.28, 0.67}
\newcommand{\set}[1]{\left\{{#1}\right\}}
\usepackage{subcaption}
\renewcommand{\geq}{\ensuremath{\geqslant}}
\renewcommand{\leq}{\ensuremath{\leqslant}}

\newcommand{\strategy}{\mathbf{p}}

\newcommand{\protocol}{\Xi}

\newcommand{\poa}{\text{PoA}}

\newcommand{\newcost}{\hat{c}}

\newcommand{\onlinealgo}{\textsc{GoWithTheFlow}}

\newcommand{\err}{\delta}

\newcommand{\distance}{D}

\newcommand{\terminal}{R}

\newcommand{\ordering}{\pi}

\newcommand{\OPT}{\texttt{OPT}} 

\newcommand{\MST}{\texttt{MST}}

    \def \N {\mathbb{N}}
    \def \R {\mathbb{R}}
    \def \DH {G_{H'}}
    \def \CC {X}
\allowdisplaybreaks

\title{Improved Price of Anarchy via Predictions}
\author[a]{Vasilis Gkatzelis\thanks{gkatz@drexel.edu}}
\author[b]{Kostas Kollias\thanks{kostaskollias@google.com}}
\author[c]{Alkmini Sgouritsa\thanks{alkmini@liv.ac.uk}}
\author[a]{Xizhi Tan\thanks{xizhi@drexel.edu}}
\affil[a]{Drexel University, Computer Science}
\affil[b]{Google Research}
\affil[c]{University of Liverpool, Computer Science}
\begin{document}
\newtheorem{claim}[theorem]{Claim}
\date{}
\maketitle

\begin{abstract}
A central goal in algorithmic game theory is to analyze the performance of decentralized multiagent systems, like communication and information networks. In the absence of a central planner who can enforce how these systems are utilized, the users can strategically interact with the system, aiming to maximize their own utility, possibly leading to very inefficient outcomes, and thus a high price of anarchy. To alleviate this issue, the system designer can use decentralized mechanisms that regulate the use of each resource (e.g., using local queuing protocols or scheduling mechanisms), but with only limited information regarding the state of the system. These information limitations have a severe impact on what such decentralized mechanisms can achieve, so most of the success stories in this literature have had to make restrictive assumptions (e.g., by either restricting the structure of the networks or the types of cost functions). 

In this paper, we overcome some of the obstacles that the literature has imposed on decentralized mechanisms, by designing mechanisms that are enhanced with predictions regarding the missing information. Specifically, inspired by the big success of the literature on ``algorithms with predictions'', we design decentralized mechanisms with predictions and evaluate their price of anarchy as a function of the prediction error, focusing on two very well-studied classes of games: scheduling games and multicast network formation games.
\end{abstract}

\newpage

\section{Introduction}
In this paper we revisit two classic decentralized resource allocation problems, scheduling games and network formation games, aiming to achieve improved price of anarchy bounds by leveraging predictions. Like many of the important games in algorithmic game theory, these two classes of games correspond to special cases of a very general model defined using a network with load-dependent cost functions. Given a graph $G=(V,E)$ and a set of $n$ users $N$, each user $i$ needs to use a path in order to connect from a source vertex $s_i\in V$ to a terminal vertex $t_i\in V$. For each edge $e\in E$, if a total load of $\ell$ players choose to use it, then this generates a cost $c_e(\ell)$, which is passed on to the players using it. The players strategically choose their paths, aiming to minimize their cost, and the performance of the induced game is evaluated using its \emph{price of anarchy}: the social cost in the ``worst'' Nash equilibria     of the game, over the optimal social cost.

For example, in the well-studied \emph{multicast network formation game} \citep{ADKTWR08,Li09,BCFM13,LL13,FHP16}, all the users share the same source $s$, and the cost functions are constant ($c_e(\ell)=c_e$ for all $e\in E$ and $\ell>0$). The initial work on this problem assumed that the cost of each edge is divided equally among its users, but this can give rise to Nash equilibria that are very inefficient, leading to a price of anarchy that grows linearly with the number of agents. If we knew exactly the set of users in advance (i.e., if we knew the set of terminals that would need to be connected to the source), then we could more carefully determine how to share the cost of each edge among its users, leading to a price of anarchy of 2~\citep{CRV10}.
However, in decentralized systems this information may not be fully known in advance, so the cost-sharing protocol may need to decide how to share the cost using only limited information. 

To better understand the impact of information limitations on the performance of cost-sharing protocols, prior work has introduced a range of models, depending on the amount of information available to the designer: i) \emph{oblivious} protocols, that are independent of the graph structure and the set of users, ii) \emph{resource-aware} protocols, that are aware of the graph structure and its cost functions, but unaware of the set of users, and iii) \emph{omniscient} protocols, that know everything about the instance at hand.\footnote{Some of the prior work also refers to oblivious protocols as ``uniform'' and to omniscient ones as ``non-uniform.''} 
For each of these information models, a long list of papers has aimed to design cost-sharing protocols that are \emph{stable} (i.e., guarantee the existence of pure Nash equilibria), and optimize the price of anarchy. However, even for resource-aware cost-sharing protocols, the results are often very pessimistic, unless we impose significant restrictions on the class of instances.

Although omniscient protocols require a possibly unrealistic amount of information, the assumption that resource-aware protocols have no information regarding the anticipated demand is unrealistic as well. Given the vast amounts of historical data that is stored and readily available, even off-the-shelf machine learning algorithms could provide a reasonable estimate regarding future demand. Therefore the severe information limitations that lead to these impossibility results may be unnecessarily pessimistic: a decentralized protocol could be augmented with some estimate regarding the future demand, and it could use this estimate as a guide for its cost-sharing decisions. 

To overcome analogous pessimistic results due to information limitations, the online algorithms literature introduced a model for designing and analyzing ``algorithms with predictions'' (see, e.g., \citep{LV18,PSK18,GP19,BCKP20,AGP20,APT21}). The goal is to design algorithms, enhanced with a prediction, that perform very well when the prediction is accurate, yet still maintain some worst-case guarantees even if it is not. The learning-augmented framework was very recently also adapted to multiagent systems involving strategic agents, giving rise to a research agenda focusing on the design of (centralized) ``mechanisms with predictions''~\citep{ABGOT22}. In this paper we extend this agenda beyond centralized systems and study the extent to which predictions can enable the design of more practical protocols for distributed multiagent systems, leading to improved price of anarchy bounds. The main question that we focus on is:

\vspace{1pt}
\begin{center}
\emph{Can decentralized protocols, enhanced with predictions, achieve improved price of\\ anarchy bounds, and how do these bounds depend on the prediction accuracy?}\end{center}

Two central notions in the literature on ``algorithms with predictions'' are \emph{consistency} and \emph{robustness}. The consistency of an algorithm (or, in our case, a protocol) is the performance guarantee that it achieves, assuming that the prediction it was provided with is accurate. Its robustness is the worst-case performance guarantee that it achieves, irrespective of the quality of the prediction. In some problems, achieving the optimal consistency needs to come at the cost of robustness, i.e., it is impossible to also simultaneously achieve the best known worst-case guarantees. Our main results in this paper provide decentralized protocols that simultaneously achieve the best-possible consistency and robustness guarantees, up to small constants.

\subsection{Our Results}
To evaluate the potential impact of predictions on the price of anarchy bounds that we can achieve, we design learning-augmented cost-sharing protocols that are enhanced with predictions regarding the demand that they should anticipate. Depending on the setting at hand, these predictions are on the volume of the demands or the locations of the terminals that they are associated with. Guided by this information, the protocols carefully adjust the cost share of each user and, even though they remain oblivious to the actual demand that appears, they achieve bounds that improve as a function of the prediction quality. In fact, we prove that our protocols simultaneously achieve the best known worst-case guarantees (robustness) and the best possible guarantees when the predictions are accurate (consistency), up to small constant factors.

\vspace{5pt}
\textbf{Games with General Cost Functions over Series-Parallel Graphs (Section~\ref{sec:seriesparallel}).}
We first consider the class of symmetric games (all the agents need to connect the same source to the same terminal, so they have the same set of strategies) over series-parallel graphs, which generalizes the classic scheduling games (which can be captured by a simple two-node graph with parallel edges). 
For each edge of the graph, we allow its cost function to be an arbitrary non-decreasing function of the number of agents using it. This is in contrast to most of the prior literature which imposes some type of structure or parameterization on the allowable cost functions (e.g., concavity, convexity, or some type of boundedness) (e.g., \citep{vFH13,CGS17,GPS21}). For this class of games, the best known price of anarchy upper bound via a resource-aware protocol is $O(n)$, and prior work has shown that without information regarding the number of users, no stable cost-sharing protocol (i.e., a protocol that admits a pure Nash equilibrium) can achieve a price of anarchy better than $O(\sqrt{n})$~\citep{CGS17}. In fact this lower bound holds even for scheduling games with capacitated constant cost functions\footnote{Given two constants, $c$ and $t$, a capacitated constant cost function is equal to $c$ as long as its input is at most $t$ and infinite otherwise.}. To overcome this obstacle, we consider the design of cost-sharing protocols that are enhanced with a, possibly erroneous, prediction regarding the total number of agents that will be using the system. 

Our main result in this setting is a cost-sharing protocol that uses the prediction, $\hat{n}$, on the number of users to achieve a price of anarchy of 4 when the predictions are correct, i.e., when $n=\hat{n}$. More surprisingly, we prove that this protocol maintains a good price of anarchy bound even if the prediction is inaccurate: if $\delta=|n-\hat{n}|$ is the prediction error, we prove a price of anarchy bound of $\min\{4(\delta+1),~ 4n\}$. In other words, when the prediction is accurate, this protocol achieves a price of anarchy of 4, while simultaneously guaranteeing a price of anarchy of $O(n)$, even if the prediction is arbitrarily inaccurate (which matches the best known worst-case price of anarchy bound, even for the special case of scheduling games). Furthermore, this bound provides a major improvement even if the prediction is not perfect, i.e., $\delta$ is positive but not too large.

To achieve this result, we first use an online algorithm to determine how many agents should be using each edge, assuming the prediction is correct. Then, our cost-sharing protocol applies carefully chosen penalties if the number of agents using it exceeds this ``threshold''. If we made these penalties arbitrarily high, this would guarantee a good outcome when the prediction is correct (no agent would want to suffer the penalty). However, such penalties could lead to very bad price of anarchy bounds if the number of agents was underpredicted, i.e., $\hat{n}<n$, since some of them would be forced to suffer these, otherwise unnecessary, penalties. On the other hand, if the penalties are not high enough, then the agents may end up exceeding the edge usage thresholds anyway, leading to inefficient outcomes and high price of anarchy, even if the prediction is correct. The main novelties of our protocol are two-fold: i) First, the way in which it determines the threshold for each edge as a function of the graph structure and the prediction $\hat{n}$, using an online algorithm (see Section~\ref{sec:online} for more details). ii) Second, the way it determines how much to penalize the agents that exceed this threshold, in order to optimize the aforementioned trade-off. 

\vspace{5pt}
\textbf{Multicast Network Formation Games over General Graphs (Section~\ref{sec:multicast}).} We then also revisit the well-studied class of multicast network formation games over general graphs. For this class of games, we know that without any information regarding the set of users, no stable cost-sharing protocol can achieve a price of anarchy better than $O(\log(n))$~\citep{CS16}. Aiming to overcome this obstacle, we turn to mechanisms that are enhanced with predictions regarding the users. However, since these are not symmetric games (each agent may want to connect to a different terminal node in the graph, so their set of strategies can be very different), knowing the \emph{number} of agents alone is not sufficient. Therefore we consider the design of cost-sharing protocols equipped with a prediction on the \emph{set of terminals} $H\subseteq V$, corresponding to the locations in the graph where the agents' terminals are expected to appear. 

In this setting, we design a cost-sharing protocol that uses the predicted terminals $H$ and achieves a price of anarchy of 4 when the predictions are correct. Crucially, as in the previous setting, our bounds provide good price of anarchy guarantees even if the predictions are inaccurate. First, as a warm-up, we assume that the set of agents is known and the predictions are regarding the location of each agent's terminal. If the distance of each terminal $t_i$ from its predicted location is $d_i$ and the overall prediction error is $D=\sum_{i\in N}d_i$, then we achieve a price of anarchy upper bound of $\min\left\{4+\frac{6D}{\OPT},~ \log n\right\}$, where $\OPT$ is the optimal social cost.  What is particularly appealing about this bound is that it maintains the best possible worst-case price of anarchy of $O(\log n)$, even if the predictions are arbitrarily inaccurate, while simultaneously guaranteeing much stronger bounds when the prediction error is small. We then move one step further and also consider settings where even the number of agents that will arrive is unknown.
In this case, we define the prediction error by generalizing a framework recently proposed in the context of online graph algorithms~\citep{APT21}. If $R$ is the set of terminals of the agents that actually appear and $H$ is the set of predicted terminals (where $|R|$ may not be equal to $|H|$), then we consider any assignment $\eta: R'\to H'$, where $R'\subseteq R$ is a subset of agent terminals and $H'\subseteq H$ is a subset of predicted locations (not all terminals need to be assigned to a prediction and multiple terminals could be assigned to the same prediction). For a given assignment $\eta$ (and the corresponding subsets $R'$ and $H'$), we let $\delta$ be the number of unassigned terminals and predicted locations (i.e., $\delta= |H\setminus H'|+|R\setminus R'|$) and let $D$ be the total distance with respect to the assignment (i.e., the sum over all the terminals $t\in R'$ of their distance from their assigned prediction $\eta(t)\in H'$). If $\mathcal{D}$ is the set of all $(D, \delta)$ pairs that correspond to some assignment $\eta$, our main result on multicast network formation games is a protocol whose price of anarchy is at most $\min\left\{\min_{(D,\delta)\in\mathcal{D}}\left\{ 4+\frac{6\distance}{\OPT}+\log \err \right\},~\log n \right\}$. Note that if $|H|=|R|$, i.e., the number of agents is predicted correctly, this bound is at least as good as the bound we achieved when the set of agents is known: we can just use the minimum weight matching of agents to predictions as the assignment $\eta$. However, our new bound can be even stronger, since we can also keep some agents unassigned, or assigned to the same prediction.

The way that our cost-sharing protocol achieves this bound is very different from the approach used in the symmetric setting of Section~\ref{sec:seriesparallel} (the fact that they both yield a price of anarchy of $4$ for perfect predictions is merely a coincidence). Using the predicted terminals, the protocol first computes the minimum spanning tree that connects these predicted terminals to the source. Then, based on the structure of this tree, it determines a priority ordering over the predicted terminals, and this ordering is extended to all other nodes as well, based on their proximity to the predicted terminals. Once this global priority ordering of all the vertices has been determined, the cost-sharing protocol is rather straightforward: the whole cost of each edge is charged to the user whose terminal has the highest priority. Therefore, the main novelty of the protocol is the way in which this global ordering is determined, using the graph structure and the predictions as input.

\subsection{Related Work}
Our work extends the literature on resource-aware cost-sharing protocols for optimizing equilibria, specifically the Price of Anarchy (PoA) and Price of Stability (PoS) metrics. The PoA measures the worst-case inefficiency of the worst equilibrium in the game whereas PoS measures the inefficiency of the best equilibrium in the game. \citet{CS16} were the first to study this family of mechanisms, focusing on the class of network formation games (like \citet{CRV10} did, from the perspective of oblivious mechanisms). 
They showed that, when the graph is outerplanar, resource-aware mechanisms can outperform oblivious ones, but they also proved that an analogous separation is not possible for general graphs. In subsequent work, \citet{CGS17} designed resource-aware mechanisms for the case of scheduling games (which correspond to special case of parallel-link graphs), and were able to achieve a constant PoA for instances with convex and concave cost functions. Subsequently, \citet{CGLS20} extended many of these results to graphs, beyond parallel links, including directed acyclic or series parallel graphs, with convex or concave cost functions on the edges.

Resource-aware cost-sharing protocols with some additional prior information regarding the users were also part of the model studied by \citet{CS16} for the case of network formation games. Specifically, rather than assuming that the source vertex of each agent is chosen adversarially, they assumed that it is drawn from a distribution over all vertices. The cost-sharing mechanism is aware of this stochastic process, so they designed a mechanism that leverages this information to achieve a constant PoA. Following-up on this work, \citet{CLS16} extended the constant PoA to include Bayesian Nash equilibria. Recently, \citet{GPS21} showed that with some information about the users the PoA of resource-aware protocols can be significantly improved for the class of scheduling games with bounded cost functions. In this work, two different types of information are considered: knowing two of the participating agents' IDs in advance, or knowing the probability with which each of the agents appears in the system.

An important characterization of the stability property for oblivious cost-sharing mechanisms was given in \citep{GMW14}. They proved that these mechanisms correspond to the class of generalized weighted Shapley values. Leveraging this characterization, \citet{GKR16} analyzed this family of cost-sharing protocols and showed that the PoA achieved by the unweighted Shapley value is optimal for a large family of network cost-sharing games.

Other papers on the design and analysis of cost-sharing protocols include the work of \citet{vFH14}, who focused on capacitated facility location games, \citet{GKK15} who proved tight bounds for general cost-sharing mechanisms, \citet{MW13}, who considered a utility maximization model, and \citet{HHHS18}, who considered a model that imposes some constraints over the portions of the cost that can be shared among the agents. Also, \citet{HM11} 
studied the performance of several cost-sharing protocols in a setting where each player can declare a different demand for each resource.

The PoS has received less attention than the PoA in terms of designing mechanisms that seek to optimize it by leveraging the network's structure or information about the participating agents. Instead, the PoS has been studied for specific classes of omniscient cost-sharing mechanisms, such as fair cost-sharing and weighted Shapley values. Beginning with the directed network formation game, \citet{ADKTWR08} proved tight logarithmic bounds for the directed network formation game with fair cost-sharing. Subsequently, \citet{KR15} showed tight PoS bounds for the class of weighted Shapley values. Various works study fair cost-sharing in the more challenging undirected model for network formation games. For broadcast games (where all nodes are terminals for players who originate at the same root) a (large) constant upper bound was given by \citet{BFM13}. For multicast games (where all players have the same root but not all nodes are player terminals) the best known upper bound is $O(\log n/\log \log n)$ \citep{Li09}, where $n$ is the number of players. In general networks, the upper bound is $O(\log n)$, which follows by \citep{ADKTWR08}. The best known lower bounds for the various models are small constants given by \citet{BCFM13}. The work of \citet{LL13} and \citet{FHP16} presented evidence that constant upper bounds are likely in multicast and general games. Going beyond network formation games, \citet{CG16} prove asymptotically tight bounds on the PoS in games with polynomial edge cost functions.

Finally, there are several other models in which cost-sharing has played a central role. For example, \citet{MS01} focused on participation games, while \citet{M08} and \citet{MR09} studied queueing
games. \citet{CGV17} recently also pointed out some connections between cost-sharing mechanisms and the literature on coordination mechanisms, which started with the work of \citet{CKN09} and led to several papers focusing on scheduling games from a designer's perspective \citep{I+09, AJM08, C09, AH12, K13, C+13, CMP14, BIKM14}. Just like the research on cost-sharing mechanisms, most of the work on coordination mechanisms studies how the PoA varies with the choice of local scheduling policies on each machine (i.e., the order in which to process jobs assigned to the same machine). 

Our work is also related to the design of learning-augmented algorithms which leverage predictions from machine-learned models. The underlying goal is to design algorithms that gracefully degrade as the prediction error increases and still achieve non-trivial worst-case guarantees. Several recent papers study optimization problems in this context. \citet{LV18} study such algorithms for the caching problem, \citet{GP19} and \citet{AGP20} focus on rent or buy, \citet{PSK18} on scheduling, \citet{BCKP20} on online learning, \citet{MV17} on reserve price optimization, and \citet{HIKV19} on frequency estimation. More recently \citet{APT21} focus on a collection of graph problems and \citet{ABGOT22} on the design of strategyproof mechanisms with predictions.

\section{Preliminaries}
\label{sec:prelims}

We consider two classes of games played on an undirected graph $G=(V,E)$ by a set of players $N = \{1,\dots,n\}$, corresponding to the set of users. In these games, each player $i\in N$ needs to choose a path in $G$ that connects a designated \emph{source} $s$ (which is the same for all players) to a \emph{terminal} $t_i$ (which may be different for each player). Each edge $e\in E$ is characterized by a cost function $c_e : \N \rightarrow \R^+$, where $c_e(\ell)$ is the cost of the edge $e$ when the load on the edge, i.e., the number of player using it, is $\ell$. The cost function for every edge $e$ satisfies $c_e(0) = 0$, i.e., no cost is induced on $e$ unless some player uses it. The first class of games that we consider is {\em symmetric series parallel network games}, where every player $i\in N$ has the same terminal $t$ and the graph $G$ is series-parallel (series-parallel graphs are defined recursively using two simple composition operations; see Section~\ref{sec:seriesparallel} for a formal definition.) In this class of games, we allow the cost function of each edge to be an arbitrary non-decreasing function. Note that this includes the well-studied class of \emph{scheduling games}, which can be captured using a multigraph with just two vertices, $s$ and $t$, and multiple parallel links connecting them (where each edge corresponds to a machine). The second class of games that we consider is {\em multicast network formation games}, where the graph can be arbitrary and each agent can have a different terminal $t_i$, but the cost $c_e(\ell)$ of each edge $e\in E$ is equal to some edge-specific constant, $c_e$, for any load $\ell\geq 1$.

\vspace{5pt}
\textbf{Strategy Profile} In all games, let $\mathcal{P}_i$ be the set of all possible strategies for player $i$, i.e., the set of paths between the vertices that player $i$ wants to connect. In the network formation games that we consider, this set can be different for each player, but the series-parallel class game is {\em symmetric}, meaning that every player has the same strategy set $\mathcal{P}_i$. In the special case of scheduling games, the strategies are single edges (singleton games) and $\mathcal{P}_i=E$ for all $i$. %The outcome of any of those games is 
A %feasible 
pure strategy profile is given by $\strategy = (p_1,  p_2 \dots, p_n)$, where $p_i \in \mathcal{P}_i$ is the path chosen by each player $i\in N$.

\vspace{5pt}
\textbf{Cost-Sharing Protocol} Let $S_e(\strategy) = \{i \in N: e \in p_i\}$ be the set of players using edge $e$ under strategy profile $\strategy$, and let $\ell_e(\strategy) = |S_e(\strategy)|$ be the \emph{load} on edge $e$. The cost of $e$ in this allocation is $c_e(\ell_e(\strategy))$, and this cost needs to be covered by the set $S_e(\strategy)$ of players using it. In this paper we design \emph{cost-sharing methods}, i.e. protocols that decide how the cost of each edge will be distributed among its users. Formally, a cost-sharing protocol $\protocol$ defines, at each strategy profile $\strategy$, a cost share $\xi_{ie}(\strategy)$ for each $i \in N$ and $e \in E$. For player $i$ with $e \notin p_i$, we have $\xi_{ie}(\strategy) = 0$, so only the players using an edge are responsible for its cost. We denote the total cost share of player $i$ in $\strategy$ as:
\[\xi_i(\strategy) = \sum_{e \in E}\xi_{ie}(\strategy).\] 
A cost-sharing protocol is \emph{budget-balanced} if for every edge $e$ and profile $\strategy$ we have $\sum_{i \in N} \xi_{ie}(\strategy) = c_e(\ell_e(\strategy))$, i.e., the cost shares that the protocol distributes to the players using an edge adds up to exactly the cost of the edge.

\vspace{5pt}
\textbf{Ordered Protocol} An {\em ordered} protocol is a priority-based cost-sharing protocol that is defined as follows. 
Given an ordering $\pi$ of the players and a strategy profile $\strategy$, the amount that the ordered protocol charges each player $i$ for each edge $e$ is $$\xi_{ie}(\strategy) = c_e(\ell_e^{<i}(\strategy)+1)-c_e(\ell_e^{<i}(\strategy)),$$ where $\ell_e^{<i}(\strategy)$ is the number of players using $e$ that precede player $i$ in order $\pi$. In other words, if we assumed that the players of $S_e(\strategy)$ arrive one at a time according to the ordering $\pi$, each player $i$ can be thought of as increasing the cost of edge $e$ by $c_e(\ell_e^{<i}(\strategy)+1)-c_e(\ell_e^{<i}(\strategy))$ and is charged that marginal cost.

\vspace{5pt}
\textbf{Classes of Games}
We aim to design protocols that yield efficient outcomes in all games within a class of games. Formally, a class of network games %scheduling(symmetric?)/multicast games 
$\mathbf{\Gamma} = (\mathcal{N}, \mathcal{G}, \mathcal{C}, \protocol)$ comprises a universe of players $\mathcal{N}$, a universe of graph $\mathcal{G}$, whose cost functions are chosen from the set $\mathcal{C}$, and a cost sharing protocol $\protocol$. A game $\Gamma \in \mathbf{\Gamma}$ then consists of a graph $G \in \mathcal{G}$ with cost functions from $\mathcal{C}$, a set of players $N \in \mathcal{N}$, and the cost sharing protocol $\protocol$.

\vspace{5pt}
\textbf{Pure Nash Equilibrium (PNE)} The goal of every player is to minimize her total cost share. Therefore, different cost-sharing protocols would lead to different classes of games and possibly very different outcomes. The efficiency of a game, thus, crucially depends on the choice of the protocol. To evaluate the performance of a cost-sharing protocol, we measure the quality of the pure Nash equilibria in the game that it induces. A strategy profile $\strategy$ is a \emph{pure Nash equilibrium} (PNE) of a game $\Gamma$ if for every player $i \in N$ who uses path $p_i$ in $\strategy$, and every alternative path $p_i'\in \mathcal{P}_i$, we have
\[\xi_i(\strategy) = \xi_i(p_i, \strategy_{-i}) \leq \xi_i(p'_i, \strategy_{-i} ),\] where $ \strategy_{-i}$ denotes the vector of strategies for all players other than $i$. This expression suggests that in a PNE no player can decrease her cost share by unilaterally deviating from path $p_i$ to $p_i'$ if all other players' strategies remain fixed. A PNE is a natural prediction regarding the outcome of the game, but not all games are guaranteed to possess a PNE. To address this issue, prior work on cost-sharing (as well as this paper) focuses on the design of \emph{stable} protocols, i.e., ones that they induce games with at least one PNE for every possible graph and set of players.

\vspace{5pt}
\textbf{Price of Anarchy (PoA)} To evaluate the efficiency of a strategy profile $\strategy$, we use the total cost $c(\strategy) = \sum_{e \in E} c_e(\ell_e(\strategy))$, and we quantify the performance of the cost-sharing protocol using the price of anarchy measure. Given a cost-sharing protocol $\protocol$, the \emph{price of anarchy} (PoA) of the induced class of games $\mathbf{\Gamma} = (\mathcal{N}, \mathcal{G}, \mathcal{C}, \protocol)$ is defined to be the worst-case ratio of equilibrium cost to optimal cost over all games in $\mathbf{\Gamma}$. Let $Eq(\Gamma)$ be the set of pure Nash equilibria and $F(\Gamma)$ be the set of all pure strategy profiles of the game $\Gamma$, then 
\[\poa(\mathbf{\Gamma}) = \sup_{\Gamma \in \mathbf{\Gamma}} \frac{\max_{\strategy \in Eq(\Gamma)} c(\strategy)}{\min_{\strategy^* \in F(\Gamma)} c(\strategy^*)}.\]

\vspace{5pt}
\textbf{Overcharging} In addition to budget-balanced protocols, we also consider mechanisms that may use overcharging. In effect, these mechanisms define a modified cost function $\newcost_e(\ell) \geq c_e(\ell)$ for all $e,\ell$ and then apply a budget-balanced protocol on these modified functions. As a result, the social cost of a given strategy profile $\strategy$ may be increased from $c(\strategy)$ to $\newcost(\strategy) = \sum_{e \in E}\newcost_e(\ell_e(\strategy))$. For these protocols, we measure the quality of the equilibria using the new costs, but we still compare their performance to the optimal solution based on the original cost functions: 
\[\poa(\mathbf{\Gamma}) = \sup_{\Gamma \in \mathbf{\Gamma}} \frac{\max_{\strategy \in Eq(\Gamma)} \newcost(\strategy)}{\min_{\strategy^* \in F(\Gamma)} c(\strategy^*)}.\]

\vspace{5pt}
\textbf{Informational Assumptions} Throughout this paper we focus on the design of \emph{resource-aware} cost-sharing protocols with \emph{predictions}. The prediction is a forecast on the set of players, specifically the cardinality of $N$ for symmetric series parallel network games and the set of terminals for multicast network games. The information available to the cost-sharing protocol of each edge is: the set of players using the edge, the structure of the network, the cost functions, and the prediction. The protocol does not know the realized set of players not using it or the strategies they have selected.

\section{Games with General Cost Functions over Series-Parallel Networks}\label{sec:seriesparallel}
In this section we study the impact of predictions in a class of symmetric games with general cost functions over series-parallel graphs. The graph has a designated source $s$ and a designated terminal $t$, and all users need to connect from $s$ to $t$. This is a significant generalization of the class of \emph{scheduling games}, which is captured by a multigraph of just two vertices $s$ and $t$ connected by multiple parallel edges (note that for any such multigraph we can construct an equivalent graph without parallel edges, where for each edge $(s,t)$ of the original multigraph we introduce a new vertex $w$ and two edges $(s,w)$ and $(w,t)$). Even for the special case of scheduling games, prior work has shown that no budget-balanced cost-sharing mechanism can achieve a PoA better than $O(\log n)$ \citep{vFH14}, even with full information (i.e., if it knows the set of users). Without information regarding the number of users, no cost-sharing mechanism can achieve a PoA better than $O(\sqrt{n})$, even with overcharging, and the best known PoA upper bound is $O(n)$~\citep{CGS17}.

Our main result in this section is a resource-aware mechanism with overcharging that does not know the number of users, but is enhanced with a prediction $\hat{n}$ regarding this number. We prove that the PoA of this mechanism is at most $\min\{4(\delta+1),~ 4n\}$, where $\delta=|n-\hat{n}|$ is the prediction error. This implies a PoA of 4 if the prediction is accurate, it maintains the best known PoA of $O(n)$ no matter how bad the prediction is, and it also provides a major improvement if the prediction is inaccurate but the error is not too large. Notably, in contrast to prior work that imposes structural restrictions on the types of cost functions considered, our results work for general non-decreasing cost functions.

To define how our mechanism shares the cost generated on each edge of the series-parallel graph, we first (in Section~\ref{sec:online}) define an online algorithm that determines how a sequence of requests should be directed through the network, if we did not know what the total number of requests (we provide more details regarding why the use of an online algorithm is desirable in that section). This algorithm, combined with a prediction $\hat{n}$ regarding the number of users, allows us to define a ``threshold'' on the number of users that we should expect on each edge of the network, assuming that the prediction is correct. Then (in Section~\ref{sec:parallel}), as a warm-up, we start by focusing on cost-sharing for the special case of parallel-link graphs. We design a cost-sharing mechanism that applies a penalty if the number of users on an edge exceeds that threshold. Then (in Section~\ref{sec:SeriesParallel}), we extend this cost-sharing idea to the more demanding class of instances involving general series-parallel graphs.

\paragraph{Series-Parallel Graphs} 
A \emph{series-parallel graph} (SPG) is constructed by performing a (not necessarily unique) sequence of \emph{series} and \emph{parallel} compositions of smaller SPGs, 
starting from the basic SPG, which is a single edge $(s,t)$. We refer to $s$ as the source and $t$ as the sink.
\begin{itemize}
    \item Given two SPGs $C_1$ and $C_2$ with sources $s_1$, $s_2$, and sinks $t_1, t_2$, we form a new SPG $C$ by merging $s_1$ and $s_2$ into one source $s$, and merging $t_1$ and $t_2$ into a new sink, $t$. This is known as the \emph{parallel composition} of $C_1$ and $C_2$.
    
    \item Given two SPGs $C_1$ and $C_2$ with sources $s_1$, $s_2$, and sinks $t_1, t_2$, we form a new SPG $C$ by merging $t_1$ and $s_2$ and letting $s=s_1$ be the new source and $t = t_2$ be the new sink. This is known as the \emph{series composition} of $C_1$ and $C_2$.
\end{itemize}

\subsection{Online Algorithm}\label{sec:online}

In this section we present an online algorithm for sequentially and myopically allocating players to paths of a series-parallel network, connecting the source to the sink. This algorithm is a crucial component for the design of our cost-sharing mechanisms both for scheduling games (parallel-link) and general series-parallel graphs. Specifically, this centralized algorithm is used as a guide regarding the outcome that our decentralized cost-sharing mechanisms aim to implement as a Nash equilibrium.
We first provide some intuition regarding the benefits of using an online algorithm as a guide, rather than directly aiming for the optimal solution.

\paragraph{The Benefits of Using an Online Algorithm}
Note that, if we trust that the prediction $\hat{n}$ will always be accurate, we can easily enforce an optimal outcome: since the mechanism knows the graph and the cost functions, it can compute the optimal strategy profile when $\hat{n}$ players are in the system, which we denote by $\OPT(\hat{n})$, and this strategy profile would determine how many players should be using each edge of the graph. If we let $\ell^*_e$ be the predicted optimal load assigned to each edge $e$, then the cost-sharing mechanism could penalize any player exceeding the $\ell^*_e$ threshold on any edge $e$ by charging them an arbitrarily large cost. As a result, if the prediction is correct, the optimal strategy profile would be the only equilibrium, leading to a PoA of 1. 

However, what if the prediction was actually inaccurate and the actual number of players $n$ is less than the predicted number of players $\hat{n}$?
In that case, the resulting Nash equilibria of the mechanism described above can be arbitrarily bad. 
For example, consider the multigraph of Figure \ref{fig:multigraph} with two vertices $s$ and $t$ and two parallel edges, connecting these two vertices. Given three costants $a_1 \ll a_2 \ll a_3$, let the cost function of the top edge be $c_1(\ell)=a_1$ if $\ell\leq \hat{n}-1$ and $c_1(\ell)=a_3$ if $\ell\geq \hat{n}$, and the cost function of the bottom edge be $c_2(\ell)=a_2$ for any $\ell\geq 1$. Then, the optimal assignment $\OPT(\hat{n})$ assuming the total number of players is $\hat{n}$ would have all players using the bottom edge, but if the actual number of players $n$ is smaller than that, they should all use the top edge. To enforce the $\OPT(\hat{n})$ outcome and achieve optimal consistency, the aforementioned protocol would heavily penalize any players using the top edge. As a result, if $n<\hat{n}$, the players would either be restricted to using the bottom edge, or they would be forced to suffer an unnecessary penalty, both of which would lead to high PoA.

\begin{figure}[h]
    \centering
\begin{tikzpicture}
\node (One) at (-3,0) [shape=circle,draw] {$s$}; 
\node (Two) at (3,0) [shape=circle,draw] {$t$};
  \node[text width=1cm] at (0.05,1.35) 
    {$c_1(\ell)$};
    \node[text width=1cm] at (0.05,-1.35) 
    {$c_2(\ell)$};
\draw [thick] (One) to [bend right=35]  (Two) ;
\draw [thick]      (One) to [bend left=35] (Two);
\end{tikzpicture}
    \caption{A simple instance exhibiting the benefits of using an online algorithm as a guide.}
    \label{fig:multigraph}
\end{figure}
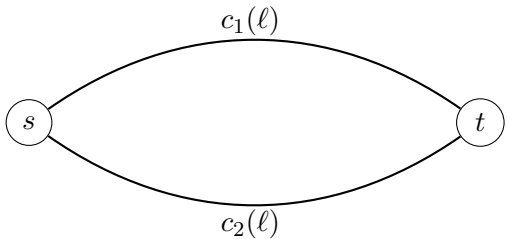

The main difficulty is that the optimal assignment $\OPT(n)$, i.e., the way in which the players are allocated to paths in this assignment, can change a lot even with small changes in $n$. As a result, a mechanism that tries to enforce the optimal assignment is bound to be very sensitive to prediction errors. To overcome this obstacle, we follow the approach of \citet{GPS21} and, rather than aiming to enforce the optimal outcome, we instead aim to enforce the outcome of an online algorithm which may be suboptimal, but is much more ``well-behaved''. Specifically, the outcome of an online algorithm does not change radically as a function of the number $n$ of arriving users: an online algorithm decides how to assign each arriving user irrevocably, without knowing how many other users would arrive in the future. Therefore, if $\ell_e(\hat{n})$ is the load that the online algorithm would assign to each edge $e$ if $\hat{n}$ players arrive, then $\ell_e(n)\leq \ell_e(\hat{n})$ for all $n\leq \hat{n}$. This way, if the actual number of players is less than what was predicted, restricting these users to a load of no more than $\ell_e(\hat{n})$ for each edge $e$ would still permit a relatively efficient assignment, namely the assignment with load $\ell_e(n)$ on each edge.

In the rest of this subsection we define an online algorithm and prove that the assignment $A(n)$ that the algorithm outputs, when $n$ users arrive, always approximates the optimal solution within a factor $4$ (Theorem~\ref{thm:onlineBound}), i.e., it has {\em competitive ratio} $\max_{G, n}\left\{c(A(n))/c(\OPT(n))\right\}=4$.

\paragraph{Online Algorithm Definition} To simplify the description of the online algorithm, without loss of generality we normalize the cost functions so that the cost of $\OPT(1)$ is 1 (this can be achieved by multiplying all cost functions by the same constant). For each $k \in \N$, let $n_k = \max\{q \in \N : c(\OPT(q)) < 2^k\}$ be the largest number of players such that the optimal social cost for assigning these players remains less than $2^k$; due to the normalization, $n_0 = \max\{q \in \N: c(\OPT(q)) = 0\}$. Using this definition, let $\ell^*_{ke}$ denote the number of players using edge $e$ in the optimal allocation when the total number of players is $n_k$. 

When the $q^{th}$ user arrives, our algorithm finds the smallest value $k$ such that for some path $p$, for all edges $e \in p$, the load so far is less than $\ell^*_{ke}$ and assigns one arbitrary of those paths to the user.
The algorithm then increments the loads $\ell_e$ for every edge $e$ on the selected path by one and moves on to the next player. A formal description is provided as Algorithm~\ref{alg:gwtf} below.

\vspace{.2cm}

\begin{algorithm}[H]\label{alg:gwtf}
$q \leftarrow 0$ \tcp*{Initialize counter for the number of players}
$\ell_e \leftarrow 0$ for each edge $e \in G$ \tcp*{Initialize all loads to zero}
\While{there exist more unassigned players}{
$q \leftarrow q +1$\\
$k_q \leftarrow \min\{k \in \N$  $|$  $\exists p \in \mathcal{P}: \ell_e < \ell^*_{ke}, \forall e \in p\}$\\
$p_q \leftarrow \mbox{ any } p \in \mathcal{P} : \ell_e < \ell^*_{k_qe}, \forall e \in p $ \tcp*{Choose any path that respects $\ell^*_{k_qe}$}
$\ell_e \leftarrow \ell_e +1$ for all $e \in p_q$ \tcp*{Assign the player to the chosen path}
}
\caption{\onlinealgo~Online Algorithm}
 \label{alg:onlinealgo}
\end{algorithm}

\vspace{.2cm}

This online algorithm, generalizes an algorithm that \citet{GPS21} introduced in the context of parallel-link graphs. The main difficulty in generalizing this algorithm beyond parallel-link graphs is the fact that in series-parallel graphs each path may need to use multiple edges to connect the source to the sink. Therefore, the feasibility of allocating the paths in an online fashion becomes significantly more complicated. In fact, as we discuss in Section~\ref{sec:conclusion}, using the same approach on a ``Braess paradox'' graph (which is just slightly more complicated than a series-paralel graph) can run into trouble. In the rest of this subsection, we first prove Lemma~\ref{lem:feasibility}, which allows us to verify the feasibility of the algorithm (Corollary~\ref{cor:kValueOfq}), and then we prove that the algorithm's competitive ratio is $4$ (Theorem~\ref{thm:onlineBound}).

\begin{lemma}\label{lem:feasibility}
Given a series-parallel graph and any $k$ and $k'$ (where $k\geq k'\geq 1$), let $O$ and $A'$ be any two assignments of $k$ and $k'$ players, respectively, to paths from $s$ to $t$ in the graph. Then, there always exists an allocation $A$ of $k$ players such that for any edge $e$,
\begin{itemize}
    \item $\ell(A_e)\geq \ell(A'_e)$ and
    \item $\ell(A_e) \leq \max\{\ell(A'_e), \ell(O_e)\}$ for any edge $e$,
\end{itemize}
where $\ell(A'_e)$, $\ell(O_e)$, and $\ell(A_e)$ are the number of players that the allocations of $A'$, $O$ and $A$, respectively, route via $e$. 
\end{lemma}

\begin{proof}
Consider $C$ be any connected component of the original graph that forms a series-parallel subgraph, and let $\ell(A'_C)$, $\ell(O_C)$, and $\ell(A_C)$ be the number of players that the allocations of $A'$, $O$, and $A$, respectively, route through $C$ (i.e., the total number of players whose assigned paths contain a sub-path connecting the source of $C$ to the sink of $C$). We will define $\ell(A_C)$ for every component $C$ (and hence for all edges) iteratively starting from the series parallel graph $G$ and following its decomposition. For $G$, we define $\ell(A_G)=\ell(O_G)=k \geq k' = \ell(A'_G)$. Suppose now that the lemma's statements are true for some component $C$, i.e., $\ell(A_C)\geq \ell(A'_C)$ and $\ell(A_C) \leq \max\{\ell(A'_C), \ell(O_C)\}$, and let $C_1$ and $C_2$ be the two components whose composition led to $C$. 
  \begin{itemize}
    \item If $C$ is constructed by the series composition of $C_1$ and $C_2$, we set $\ell(A_{C_1})=\ell(A_{C_2})=\ell(A_C)$. It is a property of series-parallel graphs that if a path passes through $C$ it should pass through both $C_1$ and $C_2$. Therefore, $\ell(A'_{C_1})=\ell(A'_{C_2})=\ell(A'_C)$ and $\ell(O_{C_1})=\ell(O_{C_2})=\ell(O_C)$. Hence, if the two statements were true for $C$, they should also be true for $C_1$ and $C_2$.
    
    \item If $C$ is constructed by the parallel composition of $C_1$ and $C_2$, $\ell(A_{C_1})$ and $\ell(A_{C_2})$ are chosen such that $$\ell(A'_{C_1})\leq \ell(A_{C_1})\leq \max\{\ell(A'_{C_1}), \ell(O_{C_1})\},$$ $$\ell(A'_{C_2})\leq \ell(A_{C_2})\leq \max\{\ell(A'_{C_2}), \ell(O_{C_2})\},$$  $$\ell(A_{C_1})+\ell(A_{C_2})=\ell(A_{C}).$$
        To verify that such values exist, we first clarify that it is a property of series-parallel graphs that if a path passes through $C$ it should pass through either $C_1$ or $C_2$. Therefore, $\ell(A'_{C})=\ell(A'_{C_1})+\ell(A'_{C_2})$ and $\ell(O_{C})=\ell(O_{C_1})+\ell(O_{C_2})$. 
    Then, note that $$\ell(A'_{C_1})+\ell(A'_{C_2})= \ell(A'_{C}) \leq \ell(A_{C}),$$ by the lemma's first statement on $C$. Additionally, 
    \begin{eqnarray*}
    \max\{\ell(A'_{C_1}), \ell(O_{C_1})\} + \max\{\ell(A'_{C_2}), \ell(O_{C_2})\} &\geq& \max\{\ell(A'_{C_1})+\ell(A'_{C_2}), \ell(O_{C_1})+\ell(O_{C_2})\}\\
    &= &\max\{\ell(A'_{C}), \ell(O_{C})\} \\
    &\geq& \ell(A_{C})\,,
    \end{eqnarray*}
    by the lemma's second statement on $C$. So, there exist values $\ell(A_{C_1})$ and $\ell(A_{C_1})$ that are between the extreme values $\ell(A'_{C_1}),\max\{\ell(A'_{C_1}), \ell(O_{C_1})\}$ and $\ell(A'_{C_2}), \max\{\ell(A'_{C_2}), \ell(O_{C_2})\}$, respectively that sum up to $\ell(A_{C})$.\qedhere.
\end{itemize}
\end{proof}

We give the following corollary of Lemma~\ref{lem:feasibility} with respect to \onlinealgo. We then use it to bound the competitive ratio of \onlinealgo.

\begin{corollary}\label{cor:kValueOfq}
For any $q\leq n_k$, the $k_q$ value computed by \onlinealgo~is at most $k$.
\end{corollary}

\begin{proof}
 Let $A'$ be the allocation of \onlinealgo~on $q-1$ players and and $O$ be the optimum allocation for $n_k$ players.  Lemma~\ref{lem:feasibility} states there exist an allocation $A$ of $n_k\geq q$ players with $\ell(A_{e'})\geq \ell(A'_{e'})$ for all edges $e'$. This means that in $A$ there exist $n_k - (q-1) \geq 1$ paths for which each edge $e$ has load $\ell(A_e)> \ell(A'_e)$. By Lemma~\ref{lem:feasibility} $\ell(A_e) \leq \max\{\ell(A'_e), \ell(O_e)\} = \ell(O_e)$. Note that it cannot be that $\max\{\ell(A'_e), \ell(O_e)\} = \ell(A'_e)$ because this would mean that $\ell(A_e)= \ell(A'_e)$. 
 Overall, we have that, given allocation $A'$, there exists at least one path $p$ such that $\ell(A'_{e}) < \ell(O_e)$ for all $e \in p$, resulting in $k_q\leq k$.
\end{proof}

\begin{theorem}
\label{thm:onlineBound}
The \onlinealgo~online algorithm for allocating flows over series parallel graphs guarantees a competitive ratio of $4$.
\end{theorem}

\begin{proof}
Let $n$ be the actual number of players, and let $k \in \N$ be the minimum value such that $n \leq n_k$. By Corollary~\ref{cor:kValueOfq}, we have that $$c(A(n)) \leq \sum_{i=1}^k \OPT(n_i) < 2^{k+1},$$
while the optimal cost is $\OPT(n) \geq 2^{k-1}$, otherwise $k$ wouldn't be minimum, leading to a competitive ratio of less than $2^{k+1}/2^{k-1} = 4$.
 \end{proof}

\subsection{Cost-Sharing Mechanism with Predictions for Games with Parallel-link}\label{sec:parallel}

We first analyze the class of parallel-link graphs as a warm-up. In a parallel-link graph, there are two vertices $s$ and $t$ and multiple parallel edges connecting them. Each player needs to pick one of these edges and a cost-sharing mechanism determines how the cost induced on each edge is to be shared among the players that use it. In deciding how to share that cost, a resource-aware mechanism knows the cost function of every edge in the graph, but not the actual number of users $n$. In this subsection, we propose a resource-aware mechanism that is enhanced with a prediction $\hat{n}$ regarding the total number of users, which allows us to achieve stronger PoA bounds.

Our cost-sharing mechanism uses the prediction $
\hat{n}$, along with the \onlinealgo~algorithm, to determine the cost that each user is responsible for. It first simulates what the \onlinealgo~algorithm would do if the number of users is, indeed, $\hat{n}$, which gives rise to the assignment $A(\hat{n})$ with a load of $\hat{\ell}_e$ on each edge $e$. Then, in any strategy profile where the number of users on an edge $e$ exceeds $\hat{\ell}_e$, the mechanism penalizes these users, so that this strategy profile will not be an equilibrium when the prediction is correct, i.e., $n=\hat{n}$. Specifically, if the prediction is correct, then there must exist another edge $e'$ whose load in this strategy profile is at most $\hat{\ell}_{e'}-1$, which means that any penalized user can unilaterally deviate to this edge and avoid the penalty.

\paragraph{Cost-Sharing Mechanism}

Let $\hat{\ell}_e$ be the load assigned to edge $e$ in the allocation $A(\hat{n})$ and let $W = c(A(\hat{n}))+\epsilon$ for some arbitrarily small $\epsilon >0$. To capture the penalties that the mechanism imposes, we define the new cost functions $\newcost$ for each edge $e$ as follows:

\begin{itemize}
    \item If $\ell_e \leq \hat{\ell}_e$ (no penalties case), then $\newcost_e(\ell_e) = c_e(\ell_e)$.
    \item If $\ell_e > \hat{\ell}_e$ (penalties case), then $\newcost_e(\ell_e) = \max\{\newcost_e(\ell_e-1) + W,~ c_e(\ell_e) )\}$.
\end{itemize}
Then, given this adjusted cost function $\newcost$ and an arbitrary ordering $\pi$ over the users, we use the ordered protocol (see Section~\ref{sec:prelims}) to share the adjusted cost of each edge among its users.\footnote{For this we may consider a global ordering $\ordering$ over the universe $\mathcal{N}$ of players. For any instance of the game with a set of players $N \in \mathcal{N}$, this global ordering implies an ordering over the players in $N$ that are actually participating. Note that this ordering $\pi$ can also be chosen at random, or it can be updated periodically to ensure a fair treatment of the players ex-ante, or in the long run, respectively.}
Note that for any edge $e$, if $\ell_e \leq \hat{\ell}_e$, the protocol is budget-balanced. On the other hand, if $\ell_e > \hat{\ell}_e$, then the last $\ell_e - \hat{\ell}_e$ users of the edge (according to $\pi$) suffer a cost at least $W$ (because this is defined as the minimum marginal increase of $\newcost$ for additional load beyond $\hat{\ell}_e$).

\begin{theorem}
Our cost-sharing mechanism is stable and, given a prediction with error $\err = |\hat{n}-n|$, it guarantees a price of anarchy of at most $\min\{4(\err+1),~4n\}$. 
\end{theorem}
\begin{proof}
We will show that for any PNE $\strategy$ in the game induced by our cost-sharing mechanism, it holds that:
$$\sum_{i \in N} \xi_i(\strategy) \leq 4(\err+1) \cdot c(\OPT(n)) \quad \text{and} \quad \sum_{i \in N} \xi_i(\strategy) \leq 4n \cdot c(\OPT(n))\, . $$
First, note that the stability of the mechanism, i.e., the PNE existence guarantee, is directly implied by the fact that we use an ordered protocol for sharing the modified cost \citep{GMW14}. We now proceed to prove the price of anarchy bound by considering two cases, depending on whether the predicted number of players is higher than the ones that actually arrived (overprediction) or lower (underprediction). 
In what follows, $\ell_e$ is the load of edge $e$ under $\strategy$.

\paragraph{Overprediction case: $\hat{n}> n$.} We first show that in $\strategy$, $\ell_e \leq \hat{\ell}_e$ for any edge $e$ (recall that $\hat{\ell}_e$ is the load assigned to edge $e$ in the allocation $A(\hat{n})$). Assume for contradiction that for some edge $e$, we have $\ell_e > \hat{\ell}_e$. This implies that there exists an edge $e'$ where $\ell_{e'} < \hat{\ell}_{e'}$, since $\hat{n} > n$. Then consider the highest ranked player $i$ on $e$, her payment is at least $W = c(A(\hat{n}))+\epsilon$. However, she can decrease her cost by moving to edge $e'$, where her charge is at most
$$c_e(\ell_e + 1) \leq c_e(\hat{\ell}_e) \leq c(A(\hat{n})) < W.$$ Since such unilateral deviation exists, the outcome can not be a PNE, which contradicts our assumption.

Now suppose that in $\strategy$ some player $i$ chooses an edge $e$ such that $\ell_e > \ell_e(A(n))$. Note that $\xi_i(\strategy) \leq c(A(n))$, otherwise she can reduce her cost by moving to another edge $e'$ such that $\ell_{e'} < \ell_{e'}(A(n))$  (there is at least one, since at least one user deviated from $A(n)$). Now since in any PNE $\ell_e \leq \hat{\ell}_e$, for any edge $e$, there are at most $\min\{\err,n\}$ deviations from $A(n)$. Therefore the total cost in $\strategy$ is,
\[\sum_{i \in N} \xi_i(\strategy) \leq c(A(n)) + \sum_{i: p_i \notin A(n)} \xi_i(\strategy) \leq (\err+1) \cdot c(A(n)) \leq 4(\err+1) \cdot c(\OPT(n)).\]

\paragraph{Underprediction case: $\hat{n} < n$.} Analogously, we first show that in $\strategy$, $A(\hat{n})$ is fully utilized, i.e., $\ell_e \geq \hat{\ell}_e$ for all $e$. Assume for contradiction, that for some edge $e$, we have $\ell_e < \hat{\ell}_e$. This implies that there is an edge $e'$ such that $\ell_{e'} > \hat{\ell}_{e'}$, since $n > \hat{n}$. Then consider the lowest ranked player $i$ on $e'$, her current payment is at least $W$, whereas if she were to move to edge $e$, her new cost would be $$c_e(\ell_e+1) \leq c_e(\hat{\ell}_e) \leq c(A(\hat{n})) < W.$$ Since such unilateral deviation exist, the outcome again cannot be a PNE which is a contradiction. 

We now consider the cost of any player $i$ with an edge $e$ such that $\ell_e > \ell_e(A(\hat{n}))$. Since $A(\hat{n})$ is fully used, we know that there are exactly $\err$ such users. Consider any such user $i$ and the edge they pick $e$. If $\ell_{e} \leq \ell_{e}(A(n))$, we have $\xi_i(\strategy) \leq \max(W, c(A(n))) \leq c(A(n))+\epsilon$. If $\ell_{e} > \ell_{e}(A(n))$, then there exist another edge $e'$ such that $\ell_{e'} < \ell_{e'}(A(n))$, since $\sum_{e \in E}\ell_{e}(A(n)) = n$. Then, we also have $\xi_i(\strategy) \leq \max(W, c(A(n))) \leq  c(A(n)) + \epsilon$, otherwise player $i$ can deviate to $e'$ and pay at most $\max(W, c(A(n)))$. Since we can make $\epsilon$ arbitrarily small, for notational simplicity we henceforth assume that $\xi_i(\strategy)\leq c(A(n))$. Therefore the total cost of any PNE is,
\[\sum_{i \in N} \xi_i(\strategy) \leq c(A(\hat{n})) + \sum_{i, p_i \notin A(\hat{n})} \xi_i(\strategy) \leq (\err+1) \cdot c(A(n)) \leq 4(\err+1) \cdot c(\OPT(n))\,,\]
where the last inequality comes from Theorem~\ref{thm:onlineBound}.

At the same time, since there are only $n$ users and we showed that each user pays less than $c(A(n))$ for any strategy she picks, we also have:
\[\sum_{i \in N} \xi_i(\strategy) \leq n \cdot c(A(n)) \leq 4n \cdot c(\OPT(n)).\qedhere\]
\end{proof}

\subsection{Cost-Sharing Mechanism with Predictions for Series-Parallel Networks}\label{sec:SeriesParallel}

Now, using the intuition developed in the more tractable case of parallel-link graphs, we extend our results to general series-parallel graphs. In doing so, we need to overcome a few non-trivial obstacles. For example, unlike the parallel-link case, the strategies of the users are now not singleton, i.e., a path may contain more than one edges. As a result, if we apply a penalty on all of the edges on some user's path, this can accumulate, leading to PoA bounds that are proportional to the length of the paths. To avoid this issue, we carefully use the structure of the graph and the flows allocated by the \onlinealgo, and ensure that the penalties are distributed across the network.

\paragraph{Cost-Sharing Mechanism} Using the \onlinealgo~algorithm, we derive the allocation $A(\hat{n})$ for $\hat{n}$ players. Let $c(A_C(\hat{n}))$ denote the cost of the allocation $A(\hat{n})$ restricted to a connected component $C$ of the graph that forms a series-parallel subgraph. We first define a constant $W_C$ for each component $C$ iteratively by following the decomposition of graph $G$. To start, we define $W_G=c(A_G(\hat{n}))+\varepsilon$ for an arbitrarily small value $\varepsilon>0$. Suppose now that for some component $C$, $W_C>c(A_C(\hat{n}))$. If $C$ is constructed by the parallel composition of $C_1$ and $C_2$, then we define $W_{C_1}=W_{C_2}=W_C$. And if $C$ is constructed by the series composition of $C_1$ and $C_2$, then we define $W_{C_i} =  W_C \cdot c(A_{C_i}(\hat{n}))/c(A_C(\hat{n}))$, for $i\in \{1,2\}$. The constant $W_C$ for each component $C$ satisfies the following two properties:
\begin{enumerate}
    \item $W_C > c(A_C(\hat{n}))$.
    \item For any path $p$ connecting the endpoints $s_C, t_C$ of $C$, $\sum_{e\in p}W_e = W_C$.
\end{enumerate}

We now update the cost functions with $W_C$. Let $\hat{\ell}_C$ be the number of players using component $C$ under allocation $A(\hat{n})$ (i.e., the number of players whose path contains a sub-path connecting the source of $C$ to the sink of $C$), where $\hat{n}$ is the prediction. Each edge $e$ is a component, so this also defined the edge load $\hat{\ell}_e$. We define the new cost functions $\hat{c}_e$ for each edge $e$ as follows:
\begin{itemize}
    \item if $\ell_e \leq \hat{\ell}_e$, $\hat{c}_e(\ell_e) = c_e(\ell_e)$
    \item if $\ell_e > \hat{\ell}_e$, $\hat{c}_e(\ell_e) = \max(\hat{c}_e(\ell_e-1) + W_e, c_e(\ell_e))$.
\end{itemize}
We then use the simple ordered protocol with the new cost function by considering an arbitrary global ordering $\pi$ of players. 

For the rest of the section we denote by $\ell_C$ the actual number of players using component $C$ under some PNE. Note that for any component $C$, if $\ell_C \leq \hat{\ell}_C$, the protocol is budget-balanced, and if $\ell_C > \hat{\ell}_C$, each extra player using $C$ pays at least $W_C$ by definition.

\begin{theorem}
Our cost-sharing mechanism for series-parallel networks is stable and, given a prediction error $\err=|n-\hat{n}|$, it guarantees a price of anarchy of at most $\min\{4(\err+1),~4n\}$. 

\end{theorem}
\begin{proof}
We will show that for any PNE $\strategy$ in the game induced by our cost-sharing mechanism, it holds that:
$$\sum_{i \in N} \xi_i(\strategy) \leq 4(\err+1) \cdot c(\OPT(n)) \quad \text{and} \quad \sum_{i \in N} \xi_i(\strategy) \leq 4n \cdot c(\OPT(n))\,. $$

First, note again that the stability of the mechanism is directly implied by the fact that we use an ordered protocol for sharing the modified cost. We now proceed to prove the price of anarchy bound by considering two cases, depending on whether the predicted number of players is higher than the ones that actually arrived (overprediction) or lower (underprediction). 

\paragraph{Overprediction case: $\hat{n} \geq n$} We first show that, for all components $C$, we have $\ell_C \leq \hat{\ell}_C$. Assume for contradiction that this is not the case and following the decomposition of $G$, let $C_1$ be the first component such that $\ell_{C_1} > \hat{\ell}_{C_1}$. For this to happen, some component $C$ is constructed by the parallel composition of $C_1$ and some $C_2$, where $\ell_C \leq \hat{\ell}_C$ and in turn $\ell_{C_2} < \hat{\ell}_{C_2}$. We use the following two lemmas to show that there exist at least one player that can decrease her cost by deviating from $C_1$ to $C_2$. 

\begin{lemma}
\label{lem:existOfOverUsedPath}
Consider any two allocations $A$ and $A'$ for $k$ players. If for some component $C$ we have $\ell(A_{C})>\ell(A'_{C})$, there exists a path $p$ from the source and sink of $C$ such that $\ell(A_e)>\ell(A'_e)$ for all $e \in p$.
\end{lemma}
\begin{proof}
First note that if $C$ is a basic SPG, i.e., if $C$ is an edge, then the statement is trivially true. Now we use mathematical induction by following the construction of $C$. Suppose that the statement holds for SPG $C_1$ and $C_2$, and consider the new SPG $C$ constructed by them:
\begin{itemize}
    \item If $C$ is constructed by the series composition of $C_1$ and $C_2$, it holds that  $\ell(A_{C_1})=\ell(A_{C_2})=\ell(A_C)$ and $\ell(A'_{C_1})=\ell(A'_{C_2})=\ell(A'_C)$. If $\ell(A_C) > \ell(A'_C)$ it must be that $\ell(A_{C_1})>\ell(A'_{C_1})$ and $\ell(A_{C_2})>\ell(A'_{C_2})$. By our induction hypothesis, there exist a path from $s_1$ to $t_1$ (the source and sink, respectively, of $C_1$) and a path from $s_2$ to $t_2$ (the source and sink, respectively, of $C_2$). Furthermore, since the series composition merges vertex $t_1$ with $s_2$, we therefore have a path $p$ from $s$ to $t$ such that $\ell(A_e) > \ell(A'_e)$ for all $e \in p$. 
    \item If $C$ is constructed by the parallel composition of $C_1$ and $C_2$, it holds that $\ell(A_{C_1})+\ell(A_{C_2})=\ell(A_{C})$ and $\ell(A'_{C_1})+\ell(A'_{C_2})=\ell(A'_{C})$. If $\ell(A_C) > \ell(A'_C)$, it cannot be that both $\ell(A_{C_1})\leq \ell(A'_{C_1})$ and $\ell(A_{C_2})\leq \ell(A'_{C_2})$. W.l.o.g. $\ell(A_{C_1}) > \ell(A'_{C_1})$, and by our induction hypothesis there exist path $p$ from $s_1$ to $t_1$ such that $\ell(A_e) > \ell(A'_e)$ for all $e \in p$, which is also connecting the source and sink for the new SPG $C$ by parallel construction.\qedhere
\end{itemize}
\end{proof}

\begin{lemma}
\label{lem:ChargesOnTheWholePath}
For any component $C$, if  $ \ell_C - \hat{\ell}_C>0$, then the $\ell_C - \hat{\ell}_C$ preceding players according to the ordering $\pi$ using $C$ are charged $W_e$ for every edge $e$ that they use in $C$.
\end{lemma}
\begin{proof}
Recall that $\ell_C$ and $\hat{\ell}_C$ are the number of players using component $C$ in some PNE and under $A(\hat{n})$, respectively. 
First note that if $C$ is a basic SPG, i.e., if $C$ is an edge, then the statement is trivially true. Now we use mathematical induction by following the construction of $C$. Suppose that the statement hold for SPG $C_1$ and $C_2$ and consider the new SPG $C$ constructed by them:
\begin{itemize}
    \item If $C$ is constructed by the series composition of $C_1$ and $C_2$, then the same set of players are using both $C_1$ and $C_2$ and $\ell_{C_1}=\ell_{C_2}=\ell_C$. So the induction hypothesis holds for $C$.
    \item  If $C$ is constructed by the parallel composition of $C_1$ and $C_2$, then each player that is charged by $W_e$ for at least one edge in $C$ is charged $W_e$ for all edges he uses in $C$. Suppose that a player $i$ that is not one of the $\ell_C - \hat{\ell}_C$ lowest priority players using $C$ is charged $W_e$ for each edge of his path in $C$. Then, the number of players that have higher priority than $i$ is strictly less than $\hat{\ell}_C$ and therefore, by Lemma~\ref{lem:existOfOverUsedPath}, there exists a path $p$ in $C$ where for each edge $e \in p$ the number of those players is strictly less than $\hat{\ell}_e$. If player $i$ deviates to $p$ he would avoid any high charge $W_e$, which contradicts the fact that $\ell_C$ is the load in some PNE. So, the induction hypothesis holds for $C$ as well.\qedhere
\end{itemize}
\end{proof}

By Lemma~\ref{lem:ChargesOnTheWholePath}, there exist a player $i$ using $C_1$ that is charged at least $W_{C_1}$ for using $C_1$ (by definition of $W_{C_1}$). By Lemma~\ref{lem:existOfOverUsedPath}, there exists a path $p$ in $C_2$ such that $\ell_e < \hat{\ell}_e$ for all $e \in p$. By the definition of $W_{C}$, we have that $W_{C_1} = W_{C} > c(A_C(\hat{n})) > c(A_{C_2}(\hat{n}))$. Since player $i$ can deviate to path $p$ and pay less, such allocation is not a PNE which is a contradiction.

Having proved that in any PNE, only paths of $A(\hat{n})$ are used, we will bound the induced cost by using the property of the \onlinealgo~algorithm: that is, if $\ell^*_C$ to be the number of players using some component $C$ in $A(n)$, $\hat{\ell}_C\geq \ell^*_C$. We define $\err_C = \hat{\ell}_C-\ell^*_C \geq 0$ to be the additional number of players in the PNE that are routed via $C$ in $A(\hat{n})$ comparing to $A(n)$; note that $\err_G=\err$. We further define $d_C=\max\{\ell_C-\ell^*_C,0\}$ to be the excess of the players using $C$ compared to $A(n)$ (obviously $d_G=0$), and $D_C$ to be the set of the $d_C$ lowest ranked players using $C$ (obviously $D_G=\emptyset$). The following lemma is a key lemma in order to conclude the proof for this case.

\begin{lemma}
\label{lem:BoundPerComponent}
For the PNE $\strategy$, let $\strategy_C$ be the allocation $\strategy$ restricted to component $C$. Then, for any component $C$, 
$$c(\strategy_C)\leq \min\{\err_C+1,\ell_C\}\cdot c(A_C(n))+\sum_{i\in D_C} \xi_i(\strategy_C).$$
\end{lemma}

\begin{proof}
 We will show the claim by mathematical induction by composing the graph $G$. First note the following two properties: 
 \begin{enumerate}
     \item If a component $C$ is constructed by the series composition of $C_1$ and $C_2$, $\ell_C=\ell_{C_1}=\ell_{C_2}$, $\err_C=\err_{C_1}=\err_{C_2}$, $d_C=d_{C_1}=d_{C_2}$ and $D_C=D_{C_1}=D_{C_2}$.
     \item If a component $C$ is constructed by the parallel composition of $C_1$ and $C_2$, $\ell_C=\ell_{C_1}+\ell_{C_2}$, $\err_C=\err_{C_1}+\err_{C_2}$.%, $\err_C \geq \err_{C_1}$ and $\err_C \geq \err_{C_2}$, since $\err_{C_1},\err_{C_2}\geq 0$.
 \end{enumerate}
 
 The statement of the lemma is true for the base case which is any single edge because for any edge $e$, $c(\strategy_e) = c_e(\ell_e)\leq c_e(\ell^*_e)+\sum_{i=1}^{d_e} (c_e(\ell^*_e+i) - c_e(\ell^*_e+(i-1))) =   c_e(\ell^*_e)+\sum_{i\in D_e} \xi_i(\strategy_e)\leq \min\{\err_e+1,\ell_e\}\cdot c(A_e(n))+\sum_{i\in D_e} \xi_i(\strategy_e)$. Suppose that the statement holds for two components $C_1$ and $C_2$. 
 
\begin{itemize}
    \item If $C$ is constructed by the series composition of $C_1$ and $C_2$, then by using the above properties,
    \begin{eqnarray*}
    c(\strategy_C) &=& c(\strategy_{C_1})+c(\strategy_{C_2})\\
    &\leq& \min\{\err_C+1,\ell_C\}\cdot (c(A_{C_1}(n))+c(A_{C_2}(n)))+\sum_{i\in D_{C}} (\xi_i(\strategy_{C_1})+\xi_i(\strategy_{C_2}))\\
    &=&\min\{\err_C+1,\ell_C\}\cdot c(A_C(n))+\sum_{i\in D_C} \xi_i(\strategy_C).
    \end{eqnarray*}
    
    \item If $C$ is constructed by the parallel composition of $C_1$ and $C_2$, suppose that $i\in D_{C_1}$ and $i\notin D_C$, meaning that player $i$ is not one of the $d_C$ lowest ranked players using $C$. Then, the number of players that have higher priority than $i$ and use $C_2$ is strictly less than $\ell^*_{C_2}$, otherwise $i\in D_C$. Therefore, by Lemma~\ref{lem:existOfOverUsedPath}, there exists a path $p'$ in $C_2$ where for each edge $e \in p'$ the number of those players is strictly less than $\ell^*_e$. If player $i$ deviated from $p_C$ to $p'$ he would be charged at most $c(A_{C_2}(n))$ for this segment (for notational simplicity we omit $\varepsilon$ here). Since $\strategy$ is a PNE $\xi_i(\strategy_{C_1}) \leq c(A_{C_2}(n))$.
    
    Similarly for every player $i$ with $i\in D_{C_2}$ and $i\notin D_C$, it holds $\xi_i(\strategy_{C_2}) \leq c(A_{C_1}(n))$. Note that $d_{C_1} \leq \err_{C_1}$ since $\ell_{C_1}\leq \hat{\ell}_{C_1}$, and $d_{C_1} \leq \ell_{C_1}$. Similarly $d_{C_2} \leq \err_{C_2}$ and $d_{C_2} \leq \ell_{C_2}$. Overall,
    
    \begin{eqnarray*}
    c(\strategy_C) &=& c(\strategy_{C_1})+c(\strategy_{C_2}) \\
    &\leq&\min\{\err_{C_1}+1,\ell_{C_1}\}\cdot c(A_{C_1}(n)) +d_{C_1}c(A_{C_2}(n)) + \sum_{i\in D_{C_1}\cap D_C}\xi_i(\strategy_{C_1})\\
    &&+\min\{\err_{C_2}+1,\ell_{C_2}\}\cdot c(A_{C_2}(n))   + d_{C_2}c(A_{C_1}(n)) + \sum_{i\in D_{C_2}\cap D_C}\xi_i(\strategy_{C_2})\\
    &\leq & \min\{\err_{C_1}+\err_{C_2}+1,\ell_{C_1}+\ell_{C_2}\}\cdot c(A_{C_1}(n))\\
    &&+ \min\{\err_{C_2}+\err_{C_1}+1,\ell_{C_2}+\ell_{C_1}\}\cdot c(A_{C_2}(n)) +\sum_{i\in D_C}\xi_i(\strategy_{C})\\
    &=&(\min\{\err_C+1,\ell_C\}\cdot c(A_C(n))+\sum_{i\in D_C} \xi_i(\strategy_C).
    \end{eqnarray*}
\end{itemize}
This completes the proof.
\end{proof}

Applying Lemma~\ref{lem:BoundPerComponent} to the whole graph $G$, and noticing that $D_G=\emptyset$, $\err_G=\err$ and $\ell_G=n$
\[\sum_{i \in N} \xi_i(\strategy) = c(\strategy) \leq \min\{\err+1,n\} \cdot c(A(n)) \leq 4\min\{\err+1,n\} \cdot c(\OPT(n)),\]
where the last inequality comes from Theorem~\ref{thm:onlineBound}.

\paragraph{Underprediction case: $\hat{n} < n$ }  Analogously, we first show that all paths in $A(\hat{n})$ are used, i.e., for all component $C$, we have $\ell_C \geq \hat{\ell}_C$. Assume for contradiction that this is not the case and following the decomposition of $G$, let $C_1$ be the first component such that $\ell_{C_1} < \hat{\ell}_{C_1}$. For this to happen, some component $C$ is constructed by the parallel composition of $C_1$ and some $C_2$, where $\ell_C > \hat{\ell}_C$ and in turns $\ell_{C_2} > \hat{\ell}_{C_2}$.
By Lemma~\ref{lem:ChargesOnTheWholePath} and Lemma~\ref{lem:existOfOverUsedPath}, there exists a player $i$ on $C_2$ and a path $p$ in $C_1$ such that player $i$ can decrease her cost by deviating to path $p$. Therefore the allocation is not a PNE which is a contradiction.

Then for any PNE $\strategy$, by Lemma~\ref{lem:ChargesOnTheWholePath}, the $\err$ lowest ranked players, let $S$ be this set of players, are charged $W_e$ for every edge $e$ they use, and the $\hat{n}$ higher ranked players use $A(\hat{n})$. Considering any player $i\in S$ and any component $C$ that $i$ uses, if $\ell_e \leq \ell^*_e$ for all edges that $i$ uses in $C$, then $i$ is charged at most $c(A_C(n))$ for the edges in component $C$ (for notational simplicity we omit $\varepsilon$ here). 

Let now $C_1$ be any maximal component that $i$ uses such that $\ell_{C_1} > \ell^*_{C_1}$. For this to happen, some component $C$ is constructed by the parallel composition of $C_1$ and some $C_2$, where $\ell_C \leq \ell^*_C$ and in turns $\ell_{C_2} < \ell^*_{C_2}$. By Lemma~\ref{lem:existOfOverUsedPath}, there exists a path $p'$ in $C_2$ such that $\ell_e<\ell^*_e$ for all $e \in p'$. If any player deviated from $C_1$ to $p'$, he would be charged at most $c(A_{C_2}(n))$. Since $\strategy$ is a PNE, overall, $\xi_i(\strategy)\leq c(A(n))$. Summing over all players,

\[ \sum_{i \in N} \xi_i(\strategy) \leq c(A(\hat{n})) + \sum_{i\in S} \xi_i(\strategy) \leq (\err+1) \cdot c(A(n)) \leq 4(\err+1) \cdot c(\OPT(n)).\]

At the same time, since there are only $n$ users and we showed that each user pays less than $c(A(n))$ for any strategy she picks, we also have:
\[\sum_{i \in N} \xi_i(\strategy) \leq n \cdot c(A(n)) \leq 4n \cdot c(\OPT(n)).\qedhere\]
\end{proof}

\section{Multicast Games: General Networks and Constant Cost Functions}
\label{sec:multicast}
We now move on to study the case of multicast network cost-sharing games which, unlike the games in the previous section, are not symmetric: each player $i$ has a given terminal vertex $t_i$ which determines the player's set of strategies, i.e., the paths from the source $s$ to $t_i$. In this regard, a complete prediction does not only forecast the number of players; it also needs to specify \emph{where} the players will appear, i.e., their terminal vertices. Prior work has shown that without any information regarding the players and terminals, the PoA of any budget-balanced resource-aware mechanism is $\Omega(\log n)$ \citep{CS16}. In this section, we design cost-sharing mechanisms that are enhanced with predictions regarding the participating player locations. More specifically, the input to the mechanism is the graph $G=(V,E)$ with edge weights $c_e, \forall e\in E$, and a subset of the vertices $H\subseteq V$ that corresponds to the predicted terminals for the players.

This mechanism achieves a constant price of anarchy when the prediction is correct (i.e., constant consistency), while maintaining a worst-case price of anarchy upper bound of $O(\log n)$ (which matches the best possible worst-case price of anarchy of any resource-aware mechanism), irrespective of how imprecise the prediction may be (i.e., asymptotically optimal robustness). In fact, we prove a general price of anarchy upper bound as a function of the prediction error, showing that the mechanism's performance gracefully transitions from a constant to $\log n$ as the error increases.

Our mechanism processes the underlying graph and the prediction and produces an ordering over all the vertices in the graph. Then an ordered cost-sharing protocol is used that follows this order. We remark that any ordered cost-sharing protocol, and therefore ours, is stable, i.e., it always admits a PNE, \citep{GMW14}.

First, as a warm-up, we consider the case where the set of players in the system is known and the predictions are regarding the location of each player’s terminal. The prediction error in this case is the aggregate metric distance of each terminal from the predicted one. 
We then extend our results to the case where even the number of players in the system is unknown. In this case, a set of vertices $H$ is the prediction of the actual set of terminals $\terminal$ that eventually appear in the system. The cardinality of $H$ and $\terminal$ may not be the same and we consider the case that each predicted point may be the prediction of any terminal from $\terminal$ or it is an outlier, meaning that it is not associated with any terminal from $\terminal$. On the other end, each terminal from $\terminal$ may be associated with one predicted point or it is an outlier, meaning that there is no predicted point for that terminal. The prediction error now is a combination of the aggregated metric distance of each terminal from the associated predicted point (if any) and the number of outliers, both predicted and actual points. This 
is a generalization of the framework recently proposed in the context of online graph algorithms \citep{APT21}.

\subsection{Warm-Up: Cost-Sharing Mechanism for a Known Set of Players}
\label{sec:multicastKnownNo}
We first consider the setting where the designer is aware of the graph, the edge costs, and the player identities, however, she only has a prediction on the location of the terminal for each player. That is, for each player $i$, the designer has a prediction $\eta(t_i)\in V$ regarding what player $i$'s terminal is (this prediction may be incorrect, i.e.,  $\eta(t_i)$ could be different than $t_i$). Let $H$ be the set of the predicted vertices and $\terminal$ be the set of the actual terminals, i.e. $H=\{\eta(t_i)\mid t_i \in \terminal\}$. To simplify our presentation, we assume that there will always be one player, player $0$, on the source $s$ who is correctly predicted. Note that this has no impact on the cost of any solution. When the number of players is correctly predicted, the main source of error is the fact that the predicted location of each player can be far away from the actual location. One natural way to quantify this error is to sum up the metric distances between the predicted and the actual location. Let $d(\cdot,\cdot)$ denotes the shortest distance on the weighted graph and $d_i=d(t_i, \eta(t_i))$. We define the distance metric error $D = \sum_{t_i\in \terminal} d_i$.  We further let $\OPT$ denote the cost of the optimal solution (which is the minimum Steiner tree on the set of terminals $R$ union the source $s$).

\paragraph{Cost-Sharing Mechanism.}Let $\MST(H)$ be a minimum spanning tree for the set of the predicted vertices $H$ on the metric closure of the graph\footnote{The metric closure of an undirected graph is the complete undirected graph on its vertex set, where the edge costs equal the shortest paths in the graph.} and $T(H)$ be an Eulerian tour after doubling the edges of $\MST(H)$. We consider an order $\pi_{\eta}$ of $H$, which is given by the order of appearance of the nodes in a DFS search on $T(H)$ starting from the source. This vertex order $\pi_{\eta}$ can be mapped to an order $\pi$ of the players, i.e. the first player according to $\pi$ has a terminal $t$ such that $\eta(t)$ is the first vertex according to $\pi_{\eta}$. In turn, the order $\pi$ defines an ordered cost-sharing protocol $\Xi$ where each player is responsible for the cost of an edge if and only if she is the first in $\pi$ among edge's users. See Figure~\ref{fig:mst} below for an illustration of the process.

\begin{figure}[h]
\centering
\begin{tikzpicture}[every node/.style={circle, draw, minimum size=.4cm},scale = 0.7]
    \node at (0,0) (s) {s};
    \node at (-2,-2) (1) {1};
    \node at (-4,-4) (2) {2};
    \node at (-6,-6) (3) {3};
    \node at (-2,-4) (4) {4};
	\node at (-3,-6) (5) {5};
    \node at (-1.5,-6) (6) {6};
    \node at (0,-4) (7) {7};
    \node at (2,-2) (8) {8};
    \node at (2,-4) (9) {9};
    \node at (4,-4) (10) {10};
    \node at (4,-6) (11) {11};
	\node at (4,-2) (12) {12};
	\draw[line width=0.05cm] (s) -- (1);
    \draw[line width=0.05cm] (s) -- (8);
    \draw[line width=0.05cm] (s) -- (12);
    \draw[line width=0.05cm] (1) -- (2);
    \draw[line width=0.05cm] (1) -- (4);
    \draw[line width=0.05cm] (1) -- (7);
    \draw[line width=0.05cm] (2) -- (3);
    \draw[line width=0.05cm] (4) -- (5);
    \draw[line width=0.05cm] (4) -- (6);
    \draw[line width=0.05cm] (8) -- (9);
    \draw[line width=0.05cm] (8) -- (10);
    \draw[line width=0.05cm] (10) -- (11);
\end{tikzpicture}
\caption{An example of a minimum spanning tree for the set of predicted vertices $H$, vertices are indexed in order of $\pi_{\eta}$.}
\label{fig:mst}
\end{figure}
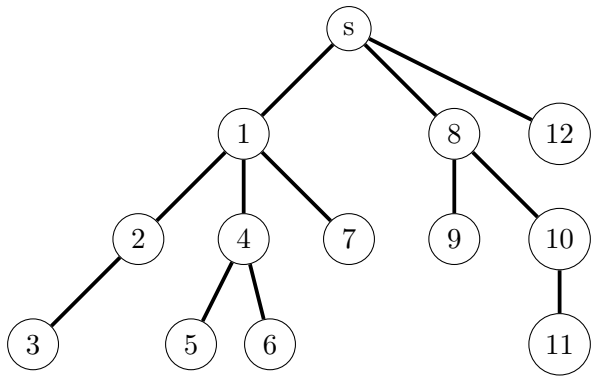 

The following theorem gives an upper bound on the price of anarchy of the above cost-sharing mechanism in this setting.

\begin{theorem}
\label{thm:multicastWarmUp}
Given a prediction with distance error $\distance$, our cost-sharing mechanism for the multicast network formation game guarantees a price of anarchy of at most $$\min\left\{ 4+\frac{6D}{\OPT},~\log n \right\}.$$
%$$\min\left\{ 4+\frac{4}{\text{OPT}}\sum_{i\in [n]} d_i,~O(\log n) \right\}.$$
\end{theorem}
\begin{proof}
First note that the $\log n $ term follows directly from the competitive ratio of the greedy algorithm for the online Steiner tree problem \citep{IW91}. The greedy algorithm takes as input some sequence of terminals and connects them one after the other to the already formed component via shortest paths. 
The competitive ratio is the comparison of the outcome of the greedy mechanism for the {\em worst case} order with the minimum Steiner tree. Our cost-sharing mechanism (in fact any ordered cost-sharing protocol) considers some sequence of the terminals given by their order. Then, it is always in the players' best interest to connect to the component already formed by previous players via shortest paths, as the greedy algorithm does.   Therefore, the competitive ratio of the greedy algorithm for the online Steiner tree problem provides an upper bound on the price of anarchy of any ordered cost-sharing protocol.
%\xnote{I think we should emphasize the story in this section, we are using predictions as a guidance to pick a mechanism out of the set of mechanisms that guarantee the best robustness.}

We next show that for any PNE $\strategy$ in the game induced by our cost-sharing mechanism, it holds that:
$$\sum_i\xi_i(\strategy) \leq 2\distance + 2(2\distance +2 \OPT) =  6\distance + 4\OPT.$$ 
For the rest of the section we use the convention that $t_0 = \eta(t_0) = s$, and we re-index the players according to the order $\pi$ computed via our mechanism.

In any PNE $\strategy$, each player $i$ would connect $t_i$ to the component formed by the players preceding her in $\pi$ via the shortest path in order to minimize her cost. Since $(t_i, \eta(t_i),\eta(t_{i-1}), t_{i-1})$ is a path connecting $t_i$ with $t_{i-1}$ (and therefore it connects $t_i$ with the component containing $s$), we can bound the cost share of player $i$ as follows
$$\xi_i(\strategy) \leq d_i + d(\eta(t_i),\eta(t_{i-1})) +d_{i-1}\,.$$
Summing over all the players we get:
\begin{eqnarray}
\sum_i\xi_i(\strategy) &\leq& 2\sum_{t_i\in R} d_i + \sum_{t_i\in R} d(\eta(t_i),\eta(t_{i-1})) \\
&\leq& 2\distance + c(T(H))\\
&\leq&  2\distance +2 c(\MST(H))\,,  \label{eq:totalcost}
\end{eqnarray}
where the inequalities comes from the fact that the total distance between the predictions are upper bounded by the cost of the Eulerian tour $c(T(H))$, which is no more than two times the cost of the minimum spanning tree $c(\MST(H))$.

We use $\OPT(H)$ to denote the optimal network, i.e., the minimum Steiner tree on $H$. Recall that $\OPT$ is the minimum Steiner tree on $R$. Consider a Steiner tree connecting $H$ formed by the union of $\OPT$ and the shortest paths between each $t_i \in R$ and $\eta(t_i)$. Then, we have:
\begin{equation}
    c(\MST(H)) \leq 2\OPT(H) \leq 2\OPT + 2D\,, \label{eq:boundMST}
\end{equation}
where the first inequality comes from the fact that any minimum spanning tree is a 2 approximation to the minimum Steiner tree. After combining \eqref{eq:totalcost} with \eqref{eq:boundMST} we get that the total cost is
\[\sum_i\xi_i(\strategy) \leq 2\distance + 2(2\distance +2 \OPT) = 6\distance + 4\OPT. \qedhere\]
\end{proof}

Note that, if the predictions are accurate, i.e., $d_i=0$ for all $i\in [n]$, the price of anarchy is at most 4. If they are inaccurate, our upper bound grows linearly as a function of the total prediction error, but it never exceeds $\log n$, even if the predictions are arbitrarily bad. Therefore, our results achieve asymptotically optimal consistency and robustness guarantees, along with a smooth transition from one to the other.

\subsection{Cost-Sharing Mechanism for an Unknown Set of Players}
\label{sec:multicastUnknownNo}
We now consider a broader prediction framework. Rather than assuming that the predicted terminals' cardinality is always correct, we consider predictions with estimates on both the number of players (terminals) and their possible locations, which may or may not be accurate. That is, the designer is given a set of $H$ corresponding to potential terminals, which may not be the same as the set of the actual terminals $R$ that appear, in terms of both size (i.e., $|R|$ may not be equal to $|H|$) and terminals' location. We adopt and generalize the framework for measuring the error in online graph algorithms defined in \citep{APT21}, which captures both error types: the number of terminals that may not predicted correctly and the distance between the predicted location and the actual terminal. In order to compute the error, we have the option to associate each actual terminal with a predicted point or to keep it ``unmatched''. Each terminal may be associated with at most one predicted point (the one that is the closest to it\footnote{If for some reason we wanted a different association, our result would still follow because the distance from any other predicted point is bounded by the one to the closest point.}), but we allow each predicted point to be associated with more than one terminal (note that this model is more general than the one considered in \citet{APT21} where each predicted point may be associated with at most one terminal). Similarly to actual points, we may keep some predicted points ``unmatched". 

Formally, let $\terminal'\subseteq \terminal$ be the set of terminals that are associated with some predicted point and $H'\subseteq H$ be the set of predicted points that are associated with at least one terminal. We denote by $\eta(t)$ the predicted point for terminal $t\in \terminal'$ and $\terminal'(u)$ the set of terminals associated with the predicted point $u \in H'$. For a given assignment $\eta$ (and the corresponding subsets $R'$ and $H'$), the prediction error is defined as a tuple $(D,\delta)$. $D$ is the distance metric error of $R'$ (defined in the same way as in the previous section) i.e., $D = \sum_{t_i\in \terminal} d_i$, where $d_i=d(t_i, \eta(t_i))$ and $d(\cdot,\cdot)$ denotes the shortest distance on the weighted graph.
$\delta$ is the number of unmatched points (terminals and predictions), i.e., $\delta = |\terminal\setminus \terminal'|+|H\setminus H'|$, those can be seen as {\em outliers}. Note that for any assignment $\eta$ we have a different tuple $(D,\delta)$; we remark that our results hold for every assignment $\eta$ and, therefore, any error pair $(D, \delta)$. See Figure~\ref{fig:matching} below for an illustration of the process.

Note that if $|H|=|R|$, i.e., the number of players is predicted correctly, this bound is at least as good as the bound we achieved when the set of players is known (Section~\ref{sec:multicastKnownNo}): we can just use the minimum weighted matching of players to predictions as the assignment $\eta$. However, our new bound can be even stronger, since we can also keep some players unassigned, or assigned to the same prediction.

\begin{figure}[h]
\centering
	\begin{tikzpicture}[every node/.style={circle, draw, minimum size=.4cm},scale=0.7]
    \node at (0,0) (s) {s};
    \node[cobalt] at (-3,-0.5) (1) {1};
    \node[cobalt] at (-4,-2) (2) {2};
    \node at (-2,-2) (3) {3};
    \node at (-4,-4) (4) {4};
    \node[cobalt] at (-6,-4) (5) {5};
    \node[cobalt] at (-8,-8) (6) {6};
    \node at (-6,-6) (7) {7};
    \node[cobalt] at (-6,-8) (8) {8};
    \node at (-2,-4) (9) {9};
    \node at (-3,-6) (10) {10};
    \node[cobalt] at (-4.5,-8) (11) {11};
    \node[cobalt] at (-3,-8) (12) {12};
    \node[cobalt] at (-1.5,-8) (13) {13};
    \node[cobalt] at (0,-8) (14) {14};
    \node at (-1.5,-6) (15) {15};
    \node at (0,-4) (16) {16};
    \node[cobalt] at (0,-6) (17) {17};
    \node at (2,-2) (18) {18};
    \node[cobalt] at (2,-6) (19) {19};
    \node at (2,-4) (20) {20};
    \node[cobalt] at (6,-4) (21) {21};
    \node at (4,-4) (22) {22};
    \node at (4,-6) (23) {23};
    \node[cobalt] at (4,-8) (24) {24};
    \node[cobalt] at (6,-6) (25) {25};
    \node at (4,-2) (26) {26};
    \node[cobalt] at (6,-0.5) (27) {27};
    \node[cobalt] at (6,-2) (28) {28};
    \draw[line width=.05cm](s) -- (3);
    \draw[line width=.05cm](s) -- (18);
    \draw[line width=.05cm](s) -- (26);
	\draw[cobalt, dashed] (1) -- (3);
	\draw[cobalt, dashed] (2) -- (3);
    \draw[line width=.05cm](3) -- (4);
    \draw[line width=.05cm](3) -- (9);
    \draw[line width=.05cm](3) -- (16);
    \draw[line width=.05cm](4) -- (7);
	\draw[cobalt, dashed] (4) -- (5);
	\draw[cobalt, dashed] (6) -- (7);
	\draw[cobalt, dashed] (7) -- (8);
    \draw[line width=.05cm](9) -- (10);
    \draw[line width=.05cm](9) -- (15);
	\draw[cobalt, dashed] (10) -- (11);
    \draw[cobalt, dashed] (10) -- (12);
    \draw[cobalt, dashed] (10) -- (13);
	\draw[cobalt, dashed] (14) -- (15);
	\draw[cobalt, dashed] (16) -- (17);
    \draw[line width=.05cm](18) -- (20);
    \draw[line width=.05cm](18) -- (22);
	\draw[cobalt, dashed] (19) -- (20);
	\draw[cobalt, dashed] (21) -- (22);
    \draw[line width=.05cm](22) -- (23);
	\draw[cobalt, dashed] (23) -- (24);
	\draw[cobalt, dashed] (23) -- (25);
	\draw[cobalt, dashed] (26) -- (27);
	\draw[cobalt, dashed] (26) -- (28);
\end{tikzpicture}
\caption{An assignment (dashed lines) of all vertices to the predicted vertex (black). The solid lines represent a minimum spanning tree on the predicted terminals. The order of the vertices that are matched with the same predicted vertex are arbitrarily assigned.}
\label{fig:matching}
\end{figure}
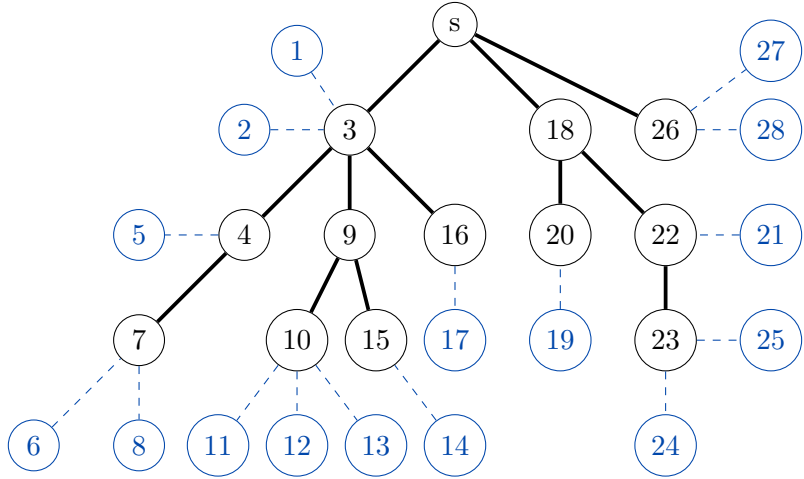

\paragraph{Cost-Sharing Mechanism.}

Let $\MST(H)$ be a minimum spanning tree for the set of the predicted vertices $H$ on the metric closure of the graph and $T(H)$ be a Eulerian tour after doubling the edges of $\MST(H)$. We consider an order $\pi_{\eta}$ of $H$ based on their first appearance in a DFS search in $T(H)$ starting from $s$. Using $\pi_{\eta}$ we define a global ordering $\pi$ of all vertices in the graph as follows. We assign each vertex $v$ to the predicted point $\eta(v)$ that is the closest to it (solve ties arbitrarily). Then we define a partial order of all vertices that respects the order of the assigned predicted points, i.e., for every two vertices $u,u'\in H$ with $u<u'$ in $\pi_{\eta}$ and every two vertices $v$ and $v'$ such that $\eta(v)=u$ and $\eta(v')=u'$, respectively, $v<v'$ in the partial order. We then define a global order $\pi$ by considering this partial order and an arbitrary order among vertices/players assigned to the same predicted point. The order $\pi$ defines an ordered cost-sharing protocol $\Xi$ where each player pays the whole cost for the edges for which she is the first in the ordering among its users.

We are now ready to give our main theorem for this setting by considering the above cost-sharing mechanism.
\begin{theorem}
If $\mathcal{D}$ is the set of all $(D, \delta)$ pairs that correspond to some assignment $\eta$, our cost-sharing mechanism for multicast network formation games guarantees price of anarchy at most $$\min\left\{\min_{(D,\delta)\in\mathcal{D}}\left\{ 4+\frac{6\distance}{\OPT}+\log \err\right\},~\log n \right\}.$$ 
\end{theorem}
\begin{proof}
Suppose any assignment $\eta$ along with the corresponding error pair $(D, \delta)$. This gives an order $\pi$ of all vertices. The $\log n $ term follows directly from \citep{IW91} as in the proof of Theorem~\ref{thm:multicastWarmUp}.

Let $\strategy$ be any PNE induced by our cost-sharing mechanism. For simplicity we denote the cost-share of each player $i$ with terminal $t$ as $\xi_t(\strategy)$ instead of $\xi_i(\strategy)$ and we denote $d_t=d(t,\eta(t))$.\footnote{We may assume without loss of generality that there exists at most one player with terminal at each vertex. If more than one players have their terminal on the same vertex, all players but the first one among them according to $\pi$ would follow the first one and be charged zero. This means that the outcome would be exactly the same even if those players weren't there.} 
In $\strategy$, each player $i$ would connect their terminal $t$ to the component formed by the players preceding her in $\pi$ via shortest path in order to minimize her cost. If $t'$ is any terminal preceding $t$ in $\pi$, $(t, \eta(t),\eta(t'), t')$ is a path connecting $t$ with $t'$ and therefore with the component containing $s$. We can bound the cost share of player $i$ as follows
\begin{equation}
\xi_t(\strategy) \leq d_t + d(\eta(t),\eta(t')) +d_{t'}\,.\label{eq:cost-share}
\end{equation}
Note that our distance error measurement $D$ captures the first and last terms in the right hand side. For the rest of the proof we aim to bound the cost of the second term.
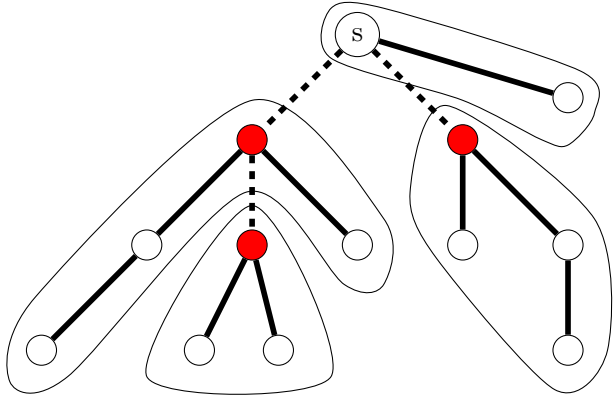
\begin{figure}[h]
    \centering
    \begin{tikzpicture}[every node/.style={circle, draw, minimum size=.4cm},scale = 0.7]
    \node at (0,0) (s) {s};
	\node[fill=red] at (-2,-2) (1) {};
    \node at (-4,-4) (2) {};
    \node at (-6,-6) (3) {};
	\node[fill=red] at (-2,-4) (4) {};
	\node at (-3,-6) (5) {};
    \node at (-1.5,-6) (6) {};
    \node at (0,-4) (7) {};
	\node[fill=red] at (2,-2) (8) {};
    \node at (2,-4) (9) {};
    \node at (4,-4) (10) {};
    \node at (4,-6) (11) {};
	\node at (4,-1.2) (12) {};
	\draw[line width=0.075cm, dashed] (s) -- (1);
    \draw[line width=0.075cm, dashed] (s) -- (8);
    \draw[line width=0.075cm] (s) -- (12);
    \draw[line width=0.075cm] (1) -- (2);
    \draw[line width=0.075cm, dashed] (1) -- (4);
    \draw[line width=0.075cm] (1) -- (7);
    \draw[line width=0.075cm] (2) -- (3);
    \draw[line width=0.075cm] (4) -- (5);
    \draw[line width=0.075cm] (4) -- (6);
    \draw[line width=0.075cm] (8) -- (9);
    \draw[line width=0.075cm] (8) -- (10);
    \draw[line width=0.075cm] (10) -- (11);

	\draw plot[smooth cycle, very thick] coordinates {(-2,-1.25) (-6.125,-5) (-6.6,-6.6) (-5.5,-6.5) (-2.125,-3) (-0.25,-4.75) (0.5,-4.75) (0.5,-3.5)};
	\draw plot[smooth cycle, very thick] coordinates {(-2,-3.25) (-4,-6.5) (-0.5,-6.5)};
	\draw plot[smooth cycle, very thick] coordinates {(0.5,0.5) (4.25,-0.7) (4.5,-1.5) (4,-2.1) (3.5,-2) (2,-1.25) (-0.5,-0.45) (-0.5, 0.5)};
	\draw plot[smooth cycle, very thick] coordinates {(1.7,-1.35) (1,-3) (1.5,-4.25) (3.75,-6.75) (4.75,-6) (4.5,-3.5)};
\end{tikzpicture}
    \caption{The figure shows the construction of $\DH$. The red, solid vertices are the set $H\setminus H'$ and the dashed lines are the ones that we remove from the $\MST(H)$ in order to get the disconnected graph $\DH$ that is formed by the four shown connected component.}
    \label{fig:DCCgraph}
\end{figure}

Starting from the $\MST(H)$, for each prediction outliers $v\in H\setminus H'$, we remove the edge between $v$ and its parent in $\MST(H)$. By doing so we separated the minimum spanning tree into at most $|H\setminus H'|$ connected components disconnected from the source $s$.See Figure~\ref{fig:DCCgraph} for an illustration of this process. We denote the discounted graph as $\DH$. 

We first bound the cost of $d(\eta(t),\eta(t'))$ for those $\eta(t)$ that are {\em not} the first according to $\pi_{\eta}$ inside the connected component that they belong to (Lemma~\ref{lem:boundPredictions}). We then would use it to bound the cost-shares of all players in $\terminal'$ that are {\em not} the first according to the global ordering of the actual terminals $\pi$ in all connected components (Claim~\ref{cl:boundNotFirsts}). For the rest $|H\setminus H'|$ players who are the first according to $\pi$ for each of the connected component\footnote{Note that for the connected component that contains the source, this player is player $0$ with terminal on the source.}, we will again use the competitive ratio of the Online Steiner tree problem \citep{IW91}.

We now focus on $\eta(t)$ that are not the first terminal according to $\pi_{\eta}$ inside their connected component. Note that the graph $\DH$ contains all the vertices of $H'$. For each connected component $\CC$ of $\DH$, let $H'(\CC) \subseteq H'$ be the vertices of $\CC$ that belongs to $H'$. Moreover, let $f_{H'(\CC)} \in H'(\CC)$ be the vertex in $\CC$ that is first according to $\pi_{\eta}$ and for every other vertex $v \in H'(\CC)$, let $p_v \in H'(\CC)$ be the vertex that precedes $v$ according to $\pi_{\eta}$ truncated to the vertices $H'(\CC)$. 

\begin{lemma}
\label{lem:boundPredictions}
$\sum_{\CC\in \DH}\sum_{v\in H'(\CC),v\neq f_{H'(\CC)}} d(v,p_v) \leq 2c(\MST(H'))\,.$
\end{lemma}
\begin{proof}
First note that each connected component $\CC$ of $\DH$ is a tree and more specifically a subtree of $\MST(H)$. Therefore, if we double the edges of $\CC$, there is an Eulerian tour $T(\CC)$ such that the first appearance of the vertices coincides with the order $\pi_{\eta}$ truncated to the vertices $H'(\CC)$. Therefore, 
$$\sum_{v\in H'_{\CC},v\neq f_{H'(\CC)}} d(v,p_v) \leq c(T(\CC)) \leq 2c(\CC).$$
Summing over all connected components of $\DH$ we get 
\begin{eqnarray}\label{eq:sumofallnormal}
\sum_{\CC\in \DH}\sum_{v\in H'(\CC),v\neq f_{H'(\CC)}} d(v,p_v) \leq 2\sum_{\CC\in \DH} c(\CC) = 2c(\DH).
\end{eqnarray}
Note that the weight of the disconnected graph $\DH$ is at most the weight of the minimum spanning tree for the set $H'$ on the metric closure of the graph, $\MST(H')$. To see this, let $E$ be the set of all edges connecting each vertex $v\in H\setminus H'$ with its parent in $\MST(H)$. Adding $E$ to $\MST(H')$ forms a spanning tree on $H$. We have:
$$c(\DH)+c(E) = c(\MST(H)) \leq c(\MST(H')) + c(E),$$
\begin{eqnarray}\label{eq:lemma5}
c(\DH) \leq c(\MST(H')).
\end{eqnarray}
Combining \eqref{eq:sumofallnormal} and \eqref{eq:lemma5} the lemma follows.
\end{proof}

We now bound the cost-share of all players in $R'$ that are {\em not} the first according to the global ordering of the actual terminals $\pi$ in all connected components. For each connected component $\CC$ of $\DH$, let $\terminal'(\CC)\subseteq \terminal'$ be the set of all terminals associated with some vertex in $H'(\CC)$, i.e., $\terminal'(\CC)=\cup_{v\in H'(\CC)}\terminal'(v)$ (recall that $\terminal'(v)$ are the terminals associated with $v\in H'$, i.e., that their closest predicted point is $v$). Moreover, let $f_{R'(\CC)}$ be the first terminal among $\terminal'(\CC)$ (or equivalent among $\terminal'(f_{H'(\CC)})$) according to $\pi$. For every other terminal $t \in \terminal'(\CC)$, let $p_t \in \terminal'(\CC)$ be the terminal that precedes $t$ according to $\pi$ truncated to the vertices $\terminal'(\CC)$. We now bound the cost-share for all terminals that are not the first according to $\pi$ in any $R'(X)$.

\begin{lemma}\label{cl:boundNotFirsts}
$\sum_{\CC\in \DH}\sum_{t\in \terminal'(\CC),t\neq f_{R'(\CC)}} \xi_t(\strategy) \leq 2c(MST(H')) + 2D\,.$
\end{lemma}
\begin{proof}
Consider any connected component $\CC$ of $\DH$, and apply inequality~\eqref{eq:cost-share} for each terminal $t \in \terminal'(\CC)$ with $t\neq f_{R'(\CC)}$ by considering $p_t$ as the terminal that appears before $t$ in $\pi$. That is, 
\begin{equation}
   \xi_t(\strategy) \leq d_t + d(\eta(t),\eta(p_t)) +d_{p_t}\,. \label{eq:cost-sharePrevious}
\end{equation}
Note that each predicted point $v\in H'(\CC)$ may be associated with many terminals, $\terminal'(v)$. For all those terminals but the first one in $\pi$, $d(\eta(t),\eta(p_t))=0$ because $\eta(t)$ and $\eta(p_t)$ are both $v$. Moreover, if $t$ is the first terminal in $\terminal'(v)$ according to $\pi$, then $p_t$ is the last terminal in $\terminal'(p_v)$ according to $\pi$. Therefore by summing over all terminals in $\terminal'(\CC)$ but $f_{R'(\CC)}$, we get:
$$\sum_{t\in \terminal'(\CC),t\neq f_{R'(\CC)}} d(\eta(t),\eta(p_t)) = \sum_{v\in H'(\CC),v\neq f_{H'(\CC)}} d(v,p_v)\,.$$
Since in inequality \eqref{eq:cost-sharePrevious} the distance of each terminal $t\in \terminal'(\CC)$ from $\eta(t)$ appears at most twice, we therefore have
$$\sum_{t\in \terminal'(\CC),t\neq f_{R'(\CC)}} \xi_t(\strategy) \leq 2 \sum_{t\in \terminal'(\CC)} d_t+\sum_{v\in H'(\CC),v\neq f_{H'(\CC)}} d(v,p_v) \,.$$
The lemma follows after summing over all connected components of $\DH$ and using Lemma~\ref{lem:boundPredictions}.\end{proof}

We now bound the cost-shares of the rest of the terminals in $\terminal'$, by using the competitive ratio of the online Steiner tree problem \citep{IW91}. Note that the optimum solution for connecting those terminals is upper bounded by $\OPT$, since those are a subset of the total set of terminals. Given that there are at most $\err_H=|H\setminus H'| $ connected components in $\DH$,
$$\sum_{\CC\in \DH}  \xi_{f_{R'(\CC)}}(\strategy) = \log \err_H \cdot \OPT\,.$$
Combining with Lemma~\ref{cl:boundNotFirsts} we get: 
\begin{equation}
\sum_{t\in \terminal'}\xi_t(\strategy) \leq \log\err_H \cdot \OPT + 2c(\MST(H')) + 2D\,.\label{eq:almostDone}
\end{equation}
Note that $c(\MST(H')) \leq 2\OPT + 2D$. To see this, we compare the minimum Steiner tree of $H'$, $\OPT(H')$, with the Steiner tree of $H$ formed by the union of $\OPT$ (the minimum Steiner tree on $R$) and the shortest path between $t$ and $\eta(t)$ for all $t \in R'$, we get:
$$c(\MST(H')) \leq 2\OPT(H') \leq 2\left(\OPT + \sum_{t\in \terminal'} d(t, \eta(t))\right)= 2\OPT + 2D,$$
Therefore \eqref{eq:almostDone} can be rewritten as:
\begin{equation}
    \sum_{t\in \terminal'} \xi_t(\strategy) \leq (\log\err_H+4) \OPT+ 6D\,.\label{eq:lem14}
\end{equation}

Similarly, we can use the competitive ratio of the online Steiner tree problem \citep{IW91} to bound the cost-shares of the terminals in $\terminal \setminus \terminal'$. For $\err_\terminal = |\terminal\setminus \terminal'|$, we have that 

$$\sum_{t\in \terminal\setminus \terminal'} \xi_t(\strategy) \leq \log\err_\terminal \cdot \OPT\,.$$
Combining with \eqref{eq:lem14} the theorem follows.
\end{proof}

\section{Conclusion and Future Directions}
\label{sec:conclusion}
In this work we extend the learning-augmented framework toward the design of decentralized mechanisms in strategic settings. This framework has recently received a lot of attention in the algorithm design literature and was very recently also extended to mechanism design. In our setting, the information that the designer is missing is due to the decentralized nature of the system, and the goal of this paper is to evaluate the extent to which predictions could overcome this obstacle. Our main results show that augmenting decentralized mechanisms with predictions can lead to major improvement in the price of anarchy bounds achievable in both scheduling games and multicast network formation games. In the first class of games, we allow general cost functions, but restrict the structure of the graph, while in the latter we allow a general graph structure and restrict the types of cost functions. The most compelling direction for future research is to bridge this gap and to evaluate the extent to which resource-aware mechanisms with predictions can achieve good PoA bounds for more general combinations of graph structures and cost functions.

A first natural direction for future research would be to extend our results on series-parallel graphs to more general graphs. However, even if we just slightly expand this family of graphs to include the famous ``Braess Paradox'' graph, the approach of Section~\ref{sec:online} runs into trouble. To explain this obstacle, we provide an example of such a graph in Figure~\ref{fig:braessparadox}, where the notation $\{1, k^2\}$ on edge $(s, a)$ implies that the cost of that edge is $1$ if a single player uses it and $k^2$ if two players use it, where $k$ is some arbitrarily large value (and the costs are defined similarly for all other edges). If we tried to design an analogous online algorithm with a bounded competitive ration for this graph, the first player to arrive would need to be assigned to the path $s\to a\to b\to t$ for a cost of 3. If not, then the cost of any alternative path would be at least $k+1$ and the competitive ratio of the algorithm for the single player case would be proportional to $k$ (and, hence, unbounded). If the algorithm commits to this assignment and a second player arrives, the algorithm would need to suffer a cost of at least $k^2$, no matter what path it chooses (since there will be at least two players using either edge $(s,a)$ or $(b,t)$). This, once again, would lead to an unbounded competitive ratio, since the optimal solution for two players would be to schedule one of them through the path $s\to a\to t$ and the other one through $s\to b\to t$, leading to a cost of $2k+2$.

\begin{figure}[h]
    \centering
    \begin{tikzpicture}[scale = 0.9]
	\node[draw,circle] at (0,0) (s) {s};
	\node[draw,circle] at (6,0) (t) {t};
	\node[draw,circle] at (3,2) (a) {a};
	\node[draw,circle] at (3,-2) (b) {b};
	\draw[->] (s) -- (a) node[midway, sloped, above] {$\set{1,k^2}$};
	\draw[->] (s) -- (b) node[midway, sloped, below] {$\set{k,k}$};
	\draw[->] (a) -- (b) node[midway, right] {$\set{1,1}$};
	\draw[->] (a) -- (t) node[midway, sloped, above] {$\set{k,k}$};
	\draw[->] (b) -- (t) node[midway, sloped, below] {$\set{1,k^2}$};
\end{tikzpicture}
    \caption{An example of the difficulty in extending our results beyond series-parallel networks.}
    \label{fig:braessparadox}
\end{figure}
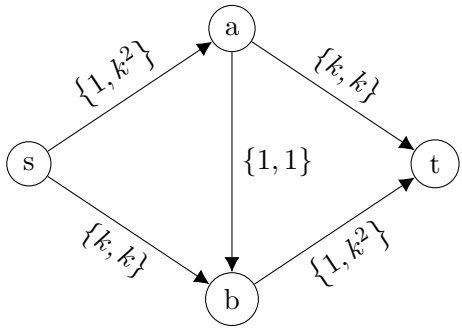

However, the main reason why we used an online algorithm as a guide to begin with is the fact that it provided a threshold $\hat{\ell}_e$ for each edge $e$ such that if we enforce this threshold as a capacity constraint on the number of players that use this edge, we can still achieve a good approximation even if the true number of agents $n$ is less than the predicted number of agents $\hat{n}$. The \onlinealgo\ algorithm goes one step further and commits to a myopic assignment of players to specific paths which is useful in our analysis (in identifying unilateral deviations in Lemma~\ref{lem:existOfOverUsedPath}), but may not be necessary. We could, instead, use the same online algorithm structure to determine these capacities without myopically committing to an assignment of players to paths. This way we can maintain the benefits discussed in Section~\ref{sec:online} without running into the obstacles that we observed for the Braess Paradox graph, above. Running this algorithm on the graph above for two players would yield a capacity of 1 on all of its edges, which would be consistent with the optimal assignment. In fact, this would maintain the competitive ratio of 4 much more broadly, but implementing its assignment as a Nash equilibrium (as we did in this paper) would require a new argument for the existence of unilateral deviations, as well as a different way of sharing the edge costs.

An alternative direction for future research would be to extend our results on multicast network formation games to the multicommodity setting, where each player $i$ may also have a different source $s_i$ (apart from a different sink $t_i$). However, we already know from \citet{CR09} that the PoA of any cost-sharing mechanism in this setting is $\Omega(\log n)$, even if the mechanism is omniscient, i.e., even if it has full information regarding the agents (which is stronger than a resource-aware mechanism with predictions). This observation points to an interesting distinction between the cost-sharing setting that we study in this paper and some recent work on graph algorithms by  \citet{APT21}: at a high level, the multicast setting that we study here bares some similarities with the online Steiner tree problem, and the multicommodity setting is analogous to the online Steiner forest problem. However, although the results of \citet{APT21} on the online Steiner tree problem are in line with the guarantees that we achieve in this paper, in their work they also achieve a constant consistency for the online Steiner forest problem, which is impossible for a cost-sharing mechanism in the multicommodity problem.

\bibliographystyle{plainnat}
\bibliography{biblio, cost-sharing}

\end{document}